\documentclass[11pt]{article}

\usepackage[top=1in, bottom=1in, left=1in, right=1in]{geometry} 
\usepackage{setspace} 
\usepackage{amsmath}
\usepackage{amsthm}
\usepackage{amssymb}
\usepackage{color}
\usepackage{bm}
\usepackage{makecell}
\usepackage{float} 
\usepackage{graphicx}

\usepackage{multicol} 
\usepackage{booktabs}
\usepackage[labelfont=bf, textfont={small,it}]{caption}
\usepackage{multirow}
\usepackage {tikz}
\usetikzlibrary {positioning}
\usetikzlibrary{matrix,decorations.pathreplacing} 
\usepackage{datatool, filecontents}
\DTLsetseparator{ = }

\definecolor{CiteGreen}{RGB}{0,160,0} 
\usepackage[colorlinks,pagebackref=true, citecolor=CiteGreen]{hyperref}

\newcommand{\say}{\emph} 

\newcommand{\E}{{\mathbb E}}

\newcommand{\N}{{\mathbb N}}
\newcommand{\R}{{\mathbb R}}
\newcommand{\Pa}{{\mathbb P}}
\newcommand{\Q}{{\mathbb Q}}

\newcommand{\Fcal}{{\mathcal F}}
\newcommand{\Scal}{{\mathcal S}}

\newcommand{\e}{\ensuremath{\mathrm{e\;\!}}} 
\newcommand{\im}{\ensuremath{\mathrm{i}}} 

\DeclareMathOperator{\spn}{span}

\DeclareMathOperator{\vect}{vec}

\newcommand{\readdata}[1]{\DTLfetch{mydata1}{thekey}{#1}{thevalue}}
\newtheorem{theorem}{Theorem}
\newtheorem*{theorem*}{Theorem}
\newtheorem{lemma}{Lemma}
\newtheorem{remark}{Remark}

\newtheorem{example}{Example}

\begin{document}
\begin{filecontents*}{data1.dat}
groundtruth_inner_sims = 100,000
groundtruth_outer_sims = 1,000,000

call_nmc_reps = 100,000
\end{filecontents*}
\DTLloaddb[noheader, keys={thekey,thevalue}]{mydata1}{data1.dat}

\title{A machine learning approach to portfolio pricing and risk management for high-dimensional problems\footnote{We thank participants at the Online Workshop on Stochastic Analysis and Hermite Sobolev Spaces, SFI Research Days, SIAM Conference on Financial Mathematics and Engineering, Virtual Risk Management and Insurance Seminar at Georgia State University, SIAM Activity Group on Financial Mathematics and Engineering virtual seminar, Ben Feng, Antoon Pelsser, and two anonymous referees for their comments.}}
\author{Lucio Fernandez-Arjona\footnote{University of Zurich. Email: lucio.fernandez.arjona@business.uzh.ch} \and
Damir Filipovi\'c\footnote{EPFL and Swiss Finance Institute. Email: damir.filipovic@epfl.ch} }
\date{21 March 2022}

\maketitle

\begin{abstract}
We present a general framework for portfolio risk management in discrete time, based on a replicating martingale. This martingale is learned from a finite sample in a supervised setting. Our method learns the features necessary for an effective low-dimensional representation, overcoming the curse of dimensionality common to function approximation in high-dimensional spaces, and applies for a wide range of model distributions. We show numerical results based on polynomial and neural network bases applied to high-dimensional Gaussian models. In these examples, both bases offer superior results to naive Monte Carlo methods and regress-now least-squares Monte Carlo.\\

\noindent \textbf{keywords:} Solvency capital; dimensionality reduction; neural networks; nested Monte Carlo; replicating portfolios.\newline
\end{abstract}

\section{Introduction}

Financial institutions face a variety of risks on their portfolios. Whether they be market and credit risk for investment portfolios, default and prepayment risk on their mortgage portfolios, or longevity risk on life insurance portfolios, the balance sheet of a bank or insurance company is exposed to many risk factors. Failure to manage these risks can lead to insolvency---with associated losses to shareholders, bondholders and/or customers---or overly conservative business strategies, which hurt consumers.

Alongside qualitative assessments---and plenty of common sense---portfolio risk management requires quantitative models that are accurate and sufficiently fast to provide useful information to management. Additionally, government regulations, such as solvency regimes, require extensive calculations to produce the required reports. Simulation techniques are often used to explore possible future outcomes. Since quantitative models require to estimate conditional expectations across a time interval, Monte Carlo simulations can be used to calculate those expectations. However, plain Monte Carlo methods suffer from problems with both accuracy and speed. 

Alternative methods have been developed over the years, some of them using functional approximation techniques. Under this approach, an approximation to the full, slow model is built using one or more faster functions. For example, the behaviour of a portfolio can be replicated via an appropriate combination of basis functions, which are faster to calculate than the original model. Most---if not all---of these alternatives suffer from several problems. Some of them are not data-driven---requiring subject matter expertise---which limits their applicability to complex problems and the ability to automate them. Others can be automated but have low quality of the approximation. And some other again are limited to low-dimensional problems.

In this paper we present a method that overcomes or greatly diminishes these problems. Our method for calculating conditional expectations uses a machine learning approach to learn a suitable function from finite samples. Therefore, the entire process is data-driven and can be free of manual steps. We show how this function can be used to obtain accurate estimates of price and risk measures, focusing on practical real-world situations in terms of runtime and number of samples being used.

The learned functions are linear combinations of functional bases, of which we present several examples, including polynomials and neural networks (of the single-layer feed-forward type). In all cases, the conditional expectations are calculated in closed-form---even for neural networks---which contributes to the accuracy and speed of the solution. We aim at high-dimensional cases, where working with a full polynomial basis is unfeasible due to the combinatorial explosion of the number of basis functions. This motivates the use of a special polynomial basis. In this basis, the input vector undergoes a data-driven, linear dimensionality reduction step---similar to a linear neural network layer---while remaining tractable with closed-form solutions, as we show. While our numerical examples are based on Gaussian measures, the method applies to a wide range of model distributions. 

In probabilistic terms, we obtain martingales that replicate the value processes---given as risk-neutral conditional expectations---of financial or insurance products. Drawing a parallel with the concept of a replicating portfolio, we also call our approach the \say{replicating martingale} method. Replicating portfolios is a widely used method in the financial industry, which relies on building linear combinations of derivatives to approximate conditional expectations necessary for market risk calculations. Our proposed replicating martingale method is also based on linear combinations of basis functions, but these are not restricted to market risk calculations. Any risk exposure can be modelled with replicating martingales because its basis functions---polynomials or neural networks---are agnostic to the underlying risk type.

Given its regression-based nature and use of simulated samples, replicating martingales is a method of the least squares Monte Carlo (LSMC) family. Within the LSMC family of methods, replicating martingales are a regress-later method (\cite{glasserman2002simulation}) given the regression is made against terminal payoffs and not against empirical conditional expectations.

We implement two numerical examples of high-dimensional Gaussian models---a plain European call option and a path-dependent life insurance product---to perform extensive testing of the accuracy of the risk calculations based on replicating martingales. In line with existing machine learning literature, we also include extensive comparisons with alternative methods, such as nested Monte Carlo and other LSMC methods, and we find our method offers superior results. Also in line with machine learning best practices, we publish the datasets for our examples, \cite{dataset}.
We hope that these datasets can be used by others to allow more direct comparisons between methods in future research.

\subsection{Related literature}

An early example of static replication via basis functions can be found in \cite{madan1994contingent}, which presents a framework for replication of general contingent claims. These contingent claims are modeled in a Hilbert space and the static replication problem is solved by constructing a countable orthonormal basis. The method is applied to the pricing and hedging of these contingent claims. \cite{carriere1996valuation} and \cite{LongstaffSchwartz} use a sequential approximation algorithm to calculate the conditional expectations required in the valuation of options with a early exercise (like American options). The method estimates these conditional expectations from the cross-sectional information in the simulation by using least squares, which gives the method the name of LSMC. \cite{andreatta2003valuing} apply this idea to valuation of life insurance policies.

The first distinction between regress-now and regress-later LSMC appears in \cite{glasserman2002simulation}. The former is the direct estimation of the conditional expectation function, and the latter the indirect estimation via regression against the terminal payoff of the contingent claim. Working in the context of American option pricing via approximate dynamic programming, they find that regress-later LSMC give less-dispersed estimates than regress-now LSMC.

The distinction and relationship between regress-now and regress-later models are important to understand our method. Regress-later models produce---ceteris paribus---better approximation functions, but introduce a few difficulties, among them the need to solve a much larger regression problem and the need to evaluate the conditional expectation of the approximation function. At the core of our paper is the demonstration of how---for polynomials and neural networks as bases---one can overcome the problem of higher dimensionality using linear dimensionality reduction, and one can calculate conditional expectations in closed form. These two factors compensate the difficulties introduced by the regress-later approach, and allow to achieve better results than a comparable regress-now approach would.

In the area of polynomial regress-now LSMC, \cite{broadie2015risk} show that such regression-based methods can---asymptotically---improve the convergence rate of nested Monte Carlo methods. They provide quality comparisons against nested Monte Carlo, and a delta-gamma approach but not against regress-later methods, whereas we do.

A regress-later model based on orthonormal piecewise linear basis functions is presented in \cite{pelsser2016difference}. Path-dependency and the resulting high-dimensional problem in long-term projections is managed via a hand-picked dimensionality-reduction function to avoid the curse of dimensionality. In that framework, the basis functions are not guaranteed to have a closed form solution, whose existence depends on the choice of dimensionality-reduction function. This function must be given to the method, and is based on expert judgement and knowledge of the problem domain. Moreover, the choice of this function implies a trade-off between complexity and dimensionality---for a given target accuracy. High-dimensional functions lead to the curse of dimensionality while low-dimensional functions might be too complex to find a closed-form solution to their conditional expectations. By contrast, our framework uses a data-driven dimensionality-reduction function in the parameter space instead of an arbitrary function.  In comparison to the piecewise linear model, which requires fixing a grid, neural networks are able to provide a data-driven grid for its activation functions.

Another application of piecewise polynomials can be found in \cite{duong2019application}. In this case, the method is based on splines in a regress-now setting. In contrast, our models are based on global polynomials in a regress-later setting.

A neural network model is applied to solvency capital problem in life insurance in \cite{castellani2018investigation}. The neural network model shows better performance than LSMC---regress-now model with polynomial basis functions---and a support vector regression model. All three models---including the neural network model---are regress-now models. By contrast, we focus on regress-later methods, which show better accuracy in the examples in addition to being better in theory.

Another approach for the calculation of risk metrics using a functional approximation in presented in \cite{bauer2020enterpriserisk}, which presents a LSMC regress-now method with a data-driven selection of basis functions, with an example given for a Gaussian model and Hermite polynomials.

The literature also contains other examples of methods not based on polynomials or neural networks, for example \cite{hong2017kernel} and \cite{risk2018sequential}. Those papers show methods to approach the same problem as the replicating martingales in this paper, namely the calculation of risk metrics for risk management purposes, but using kernel methods and Gaussian process regression, respectively.

Table \ref{tab:paper_comparison} summarizes the above discussion and gives an overview of the related literature on regress-now and regress-later methods.

\begin{table}[ht]
\small
\caption{Comparison of regress-now and regress-later methods in the literature. LDR stands for linear dimensionality reduction.}
\label{tab:paper_comparison}
\begin{tabular}{p{0.15\linewidth}p{0.25\linewidth}p{0.5\linewidth}}
\toprule
    & \textbf{Dimensionality } & \textbf{Regression method}         \\
    & \textbf{reduction on } &         \\
    & \textbf{inputs} &         \\
\midrule
\textbf{Regress-now}               &                                   &    \\[2ex]

\cite{broadie2015risk}             & none                              & incomplete, manually-selected monomial basis                                               \\[2ex]
\cite{duong2019application}        & none                              & polynomial spline space                                                                     \\[2ex]
\cite{castellani2018investigation} & implicit in neural network        & neural network                                                                              \\[2ex]
\cite{bauer2020enterpriserisk}     & none                              & left singular functions of the conditional expectation operator               \\[2ex]
\cite{hong2017kernel}              & model decomposition              & kernel smoothing                                                                            \\[2ex]
\cite{risk2018sequential}          & none                             & Gaussian process regression                                                                 \\
\midrule
\textbf{Regress-later}             &                                   &                                                                                             \\[2ex]
\cite{pelsser2016difference}       & manually selected function    & orthonormal piecewise linear basis and Hermite polynomials \\[2ex]
this paper            & LDR   & orthogonal polynomials and shallow neural networks                                                       
\end{tabular}

\end{table}

Another strand of related literature is in the field of uncertainty quantification, where polynomial surrogate functions have been used for a long time to reduce the runtime of complex models. Recently, \cite{hokanson2018data} showed that polynomial ridge approximation can be extended to reduce dimensionality in a data-driven manner. We follow a similar approach and apply it to portfolio pricing and risk management.

Nested Monte Carlo methods for portfolio management have been studied in \cite{lee2003computing} and \cite{gordy2010nested} among others. \cite{gordy2010nested} show several methods that can reduce the computational cost of nested Monte Carlo for a homogeneous portfolio of financial instruments with an additive structure. Since we aim to cover portfolios do not exhibit such additive structure---such as the life insurance portfolio---we restrict our comparisons to standard nested Monte Carlo.

An unfortunate aspect of this literature is that there are no standard models on which to compare the quality of the surrogate models, as is the case in the field of machine learning with datasets like MNIST \cite{lecun1998gradient}. Such a dataset would allow more straightforward comparisons among the advanced methods mentioned above. As in the large majority of the above literature, we compare the replicating martingale approach to some well established methods---in our case nested Monte Carlo and regress-now LSMC---but not to any of the other advanced methods proposed in the literature. In this sense, we do not claim any general superiority of our approach.

The remainder of the paper is as follows. Section~\ref{sec_RMP} formalizes the replicating martingale problem and recalls the standard nested Monte Carlo approach, which is then illustrated by means of a preliminary example in Section~\ref{sec:call_nmc_new}. Section~\ref{sec:ml_approach} describes our machine learning approach to the replicating martingale problem. It contains novel, rigorous results on the existence and uniqueness of the optimal surrogate function. Sections~\ref{sec:call_example} and \ref{sec:insurance} provide numerical case studies: the initial European call option example revisited in Section~\ref{sec:call_example}, and an insurance liability model in Section~\ref{sec:insurance}. Section~\ref{sec_conc} concludes. The appendix contains proofs and technical background material. Appendix~\ref{sec:qual_metrics} describes the quality metrics used to compare different methods. Appendix~\ref{ESG app} contains the economic scenario generator underlying the numerical examples in the main text. Appendix~\ref{sec:proof_thmVk} contains all proofs and some auxiliary results of independent interest. A comparison of the runtimes of the methods is given in Appendix~\ref{app_runtime}. Appendix~\ref{sec:sensitivity} presents an analysis of the sensitivity of the proposed methods to different hyper-parameters.

\section{The replicating martingale problem}\label{sec_RMP}

We consider an economic scenario generator with a finite time horizon $T$, where time is in units of years. Randomness is generated by an $\mathbb{R}^d$-valued stochastic driver process $X = (X_1,\dots,X_T)$ with mutually independent components $X_t$. We denote by $\Q$ the distribution of $X$ on the path space $\Omega=\R^{dT}$. The flow of information is modeled by the filtration $\Fcal_t=\sigma(X_1,\dots,X_t)$, $t=1,\dots,T$, generated by $X$. If not otherwise stated, all financial values and cash flows are discounted by some numeraire, e.g., the cash account, and we assume that $\Q$ is the corresponding risk-neutral pricing measure.

Our objective is a portfolio of assets and liabilities whose present value is to be derived from its cash flow, which accumulates to a terminal value at $T$ given as function $f$ of $X$,
\[    f(X) =   \sum_{t=1}^T  \zeta_t ,   \]
with $\Fcal_t$-measurable time-$t$ cash flows $\zeta_t  = \zeta_t(X_1 ,\dots,X_t )$. As is the case in practice, we assume that all $\zeta_t$, and thus $f$, are exogenously given functions in $L^2_\Q$. Our goal is to find the cum-dividend value process of the portfolio, given as
\[  V_t = \E^\Q_t[f(X)] = \underbrace{\sum_{s=1}^t \zeta_s}_{\text{accumulated cash flow at $t$}} + \underbrace{\E^\Q_t\bigg[\sum_{s=t+1}^T \zeta_s\bigg]}_{\text{time-$t$ spot value}} ,\]
where $\E^\Q_t[\cdot]=\E^\Q[\cdot\mid\Fcal_t]$ denotes the $\Fcal_t$-conditional expectation. Formally speaking, $V_t$ is the $L^2_{\Q}$-martingale that replicates the terminal value $f(X)$.

There are many examples that fit this description, including complex financial derivatives, insurance liabilities, mortgage-backed instruments and other structured products. In most real-world cases, $V_t$ is not given in closed form but has to be estimated from simulating $f(X)$. This creates computational challenges, as the function $f$ may be costly to query and available computational budget is limited.

There are several risk management applications of the portfolio value process. We focus here on risk measurement. Insurance and banking solvency regulatory frameworks, such as Solvency II, Swiss Solvency Test, and Basel III, require capital calculations that are based on risk measurements of the value changes $\Delta V_t=V_t - V_{t-1}$ over risk periods $(t-1,t]$. The most common risk measures are \emph{value at risk} $\mbox{VaR} _{\alpha }(L)$, defined as the left $\alpha$-quantile of the distribution of the loss $L$, for some confidence level $\alpha \in (0,1)$, and \emph{expected shortfall} $\mbox{ES}_{\alpha }(L)=\frac {1}{1-\alpha }\int_{1-\alpha}^{1}\mbox{VaR}_{\gamma }(L)d\gamma$, see, e.g., \cite[Section~4.4]{foe_sch_04}. In insurance regulation, risk is usually measured for a one-year risk horizon and the economic capital is determined by $\rho[-\Delta V_1]$, where $\rho$ is a placeholder for either $\mbox{VaR} _{\alpha }$ or $\mbox{ES}_{\alpha }$. There are cases, however, where calculations require three-year capital projections which would be calculated by $\rho[-\Delta V_1]$, $\rho[-\Delta V_2]$ and $  \rho[-\Delta V_3]$ for years one, two and three respectively, which would require the joint distribution of (parts of) the entire process $V$.

Monte Carlo simulation is the standard method for computing risk measures. We refer to as standard \say{nested Monte Carlo} to the method that is illustrated in Figure~\ref{fig:nmc}. The outer stage of the simulation consists of a set of $n_0$ simulations, $X^{(i)}_{1:t}$, $i=1,\dots, n_0$, that are independent and identically distributed as the stochastic driver $X$ up to the risk horizon time $t$. The portfolio value $V_t^{(i)}$ for each outer scenario $X^{(i)}_{1:t}$ is estimated via the inner stage of Monte Carlo simulation by taking the sample mean of the $n_1$ inner simulations drawn for each outer simulation $X^{(i)}_{1:t}$.  Once the empirical distribution for $V_t^{(i)}$ has been obtained, the desired risk measure is approximated via its empirical equivalent. While straightforward to implement, standard nested Monte Carlo leads to an exponential increase of simulations if applied to various risk time horizons. Hence it often cannot be in used in practice due to the large computational effort required. We illustrate these issues with the following initial example of a European call option.

\begin{figure}[ht]
\begin{center}\includegraphics[width=0.5\textwidth]{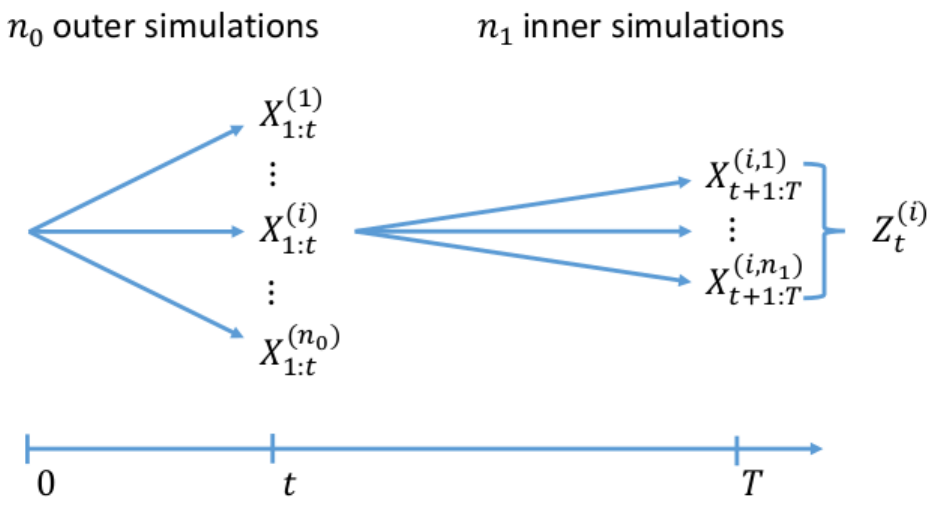}\end{center}
\caption{Nested Monte Carlo structure. The outer simulations span the time period 0 to $t$, and each of those serves as the starting point of a group of inner simulations.}
\label{fig:nmc}
\end{figure}

\section{Nested Monte Carlo for a European call option} \label{sec:call_nmc_new}

By way of preliminary example, we calculate the present value $V_0$ and expected shortfall of the one-year loss $-\Delta V_1$ for a European call on an equity index. We assume that we hold a short position in this call and therefore focus on the loss-making tail, that is, the tail where the equity index values are higher. 

The economic scenario generator is described in Appendix~\ref{ESG app}. That generator maps $X$ to a vector of economic factors. This vector contains several components, among them the equity index ${S}_t$ and the cash account $C_t$. The European call payoff at time $T$ is $\max({S}_T-K, 0)$ where $K$ is the strike of the option and $T$ its maturity. The variables ${S}_t$, $C_t$, and $K$ represent nominal---undiscounted---values. Hence for this portfolio there is one discounted cash flow at $T$, so that the terminal value function is \[f(X)=-\max({S}_T-K, 0)/C_T.\] 

In this example we work with two maturities $T=5$ and $T=40$ and for both strike price $K={S}_0=100$. The model has 3 stochastic drivers, $d=3$, and therefore there are 15 and 120 total dimensions---individual stochastic variables $X$---for $T=5$ and $T=40$, respectively.

To establish the \say{benchmark value}---the closest we can get to the \say{ground truth} without using a closed form solution\footnote{There is a closed-form solution for this particular example but our intention is to exemplify the general case, which does not have one.}---we first run a very large nested Monte Carlo simulation, with \readdata{groundtruth_outer_sims} outer simulations ($n_0$) and \readdata{groundtruth_inner_sims} inner simulations ($n_1$). We then calculate the 99\% expected shortfall on the loss-making tail. When working with an empirical distribution, as we do here based on simulation data, the expected shortfall reduces to simply averaging the 1\% worst results.

Before we can test the quality of the nested Monte Carlo estimator for a finite simulation budget, we need to decide how to split said budget between inner and outer simulations. Different combinations of outer and inner simulations will produce different nested Monte Carlo estimators. The bias and variance of the nested Monte Carlo estimator depends on both the amount of outer and inner simulations. Each combination has a different bias and variance and therefore a different mean absolute error. This is shown in Table \ref{tab:call_mae_es_nmc} for a fixed total budget of $n_0\times n_1=50{,}000$ simulations.\footnote{In keeping with industry convention, the present value is reported as positive
number, $|V_0|$, without taking into account the fact that it is a short position.} In each individual estimation (inner--outer combination), the error is calculated as a percentage of the benchmark value and therefore we refer to the quality metric as \say{MApE}, for Mean Absolute percentage Error. The formula is described in Appendix~\ref{sec:qual_metrics}. It is important to note that since the expected shortfall is calculated on $-\Delta V_1$, the error on $V_0$ is also part of the error on $\mbox{ES}_{{99\%} }(-\Delta V_1)$.

\begin{table}[ht]
    \centering
    \small
    \caption{Nested Monte Carlo, comparison of expected shortfall for nested Monte Carlo for various splits of total $n_0\times n_1=50{,}000$ simulations, where inner simulations $n_1$ range from 1 to 500.}
    \label{tab:call_mae_es_nmc}

\begin{tabular}{@{}rrrrrrrrrr@{}}
\toprule
                  & \textbf{Benchmark} & \multicolumn{8}{l}{\textbf{Nested Monte Carlo (MApE by inner simulations)}} \\
\textbf{Maturity} &   \textbf{(value)} &                                                       1 &      10 &      25 &      50 &    100 &    250 &    400 &    500 \\
\midrule
                5 &            56.8588 &                                            233.4\% &  29.5\% &  11.8\% &   7.9\% &  9.3\% & 14.4\% & 19.0\% & 20.4\% \\
               40 &            62.6205 &                                           2027.0\% & 463.6\% & 232.1\% & 129.2\% & 66.6\% & 24.6\% & 21.4\% & 19.1\% \\
\bottomrule
\end{tabular}
\end{table}

As described in \cite{broadie2015risk}, it is not possible, in general cases, to decide for a finite budget how to make this inner-outer split in an optimal way. In this paper we will err on the side of presenting optimistic risk figures for nested Monte Carlo estimations, by choosing an optimal split corresponding to the smallest MApE ES. This bias towards more accurate nested Monte Carlo risk estimations than possible in practice will not be a problem for our analysis, since we find that the proposed replicating martingale method produces more accurate results than the optimal nested Monte Carlo, which is already better than what one would obtain in practice. Note that this problem does not exist for regression-based methods, since it is not necessary to split the training budget.

Table~\ref{tab:call_mae_nmc} shows the MApE for present value and ES of the optimal combination for varying total sample size $n_0\times n_1$. We can see that the estimation of $\mbox{ES}_{{99\%} }(-\Delta V_1)$, performed via nested Monte Carlo, is much more affected by the sample size than the estimation of $V_0$. In fact, for the estimation of $V_0$ it would be optimal not to split the simulations. For example, the MApE for the present value for sample size 50,000 would be 0.6\% and 1\%, as opposed to 1.7\% and 2.4\% as shown in the last row, respectively.

\begin{table}[ht]
    \centering
    \small
    \caption{Nested Monte Carlo, comparison of present value and expected shortfall MApE (in percentage points) for optimal inner--outer split $n_0\times n_1$, for total sample size varying from 1,000 to 50,000.}
    \label{tab:call_mae_nmc}
\begin{tabular}{@{}rrrrr@{}}
\toprule
                 & \multicolumn{2}{r}{\textbf{Present Value}} & \multicolumn{2}{r}{\textbf{Expected Shortfall}} \\
\textbf{Samples} &            Maturity: 5 & Maturity: 40 &                 Maturity: 5 & Maturity: 40 \\
\midrule
           1,000 &                    6.5 &          7.5 &                        26.7 &        403.6 \\
           5,000 &                    4.0 &          3.9 &                        15.9 &        110.9 \\
          10,000 &                    3.8 &          3.1 &                        14.5 &         56.6 \\
          50,000 &                    1.7 &          2.4 &                         7.9 &         19.1 \\
\bottomrule
\end{tabular}

\end{table}

These Monte Carlo results will be used throughout the paper as one of the reference methods against which we measure our approach. Even if the limitations of standard nested Monte Carlo mean that it is not the main approach used by practitioners in large-scale problems, it remains the simplest way to approach the estimation of conditional expectations, and it provides a common baseline that both practitioners and academics can easily understand.

\section{Machine learning approach} \label{sec:ml_approach}

We now present our method, which addresses the computational challenges described and illustrated by the above example. Thereto, in a first step, we directly approximate the terminal value function $f$ by projecting it on a finite-dimensional subspace in $L^2_\Q$ that is spanned by an optimally chosen set of basis functions in $X$, which constitute the feature map. We assume that these basis functions admit conditional expectations in closed form. We thus obtain, as second step, an approximation of the portfolio value process in closed form. In practical applications, we learn the approximation of $f$ from a finite sample of $X$, which induces an empirical measure that proxies the model population measure $\Q$. The performance of this approach hinges on the choice of the basis functions. We formalize and discuss all this in detail in the following.

\subsection{Finite-dimensional approximation}

Fix a dimension $m\in\N$, and let $\Theta$ be a parameter set such that, for every $\theta\in\Theta$, there are functions $\phi_{\theta,i}:\R^{dT}\to\R$ in $L^2_\Q$, for $i=1,\dots,m$. These functions form the \say{feature map}  $\phi_\theta =(\phi_{\theta,1},\dots,\phi_{\theta,m})^\top$. For any $\theta\in\Theta$, the $L^2_\Q$-projection of $f$ on $\spn\{\phi_{\theta,1},\dots,\phi_{\theta,m}\}$ is given by $f_\theta=\sum_{i=1}^m \phi_{\theta,i} \beta_{\theta,i}=\phi_\theta^\top\beta_\theta$, where $\beta=\beta_\theta\in\R^m$ solves
\begin{equation}\label{eqnbetaAnew}
   \min_{\beta\in \R^m} \| f -\phi_\theta^\top\beta\|_{L^2_\Q} .
\end{equation}

We assume that the conditional expectations $\E^\Q_t[\phi_{\theta}(X)]=G_{\theta,t}(X_1,\dots,X_t)$ are given in \say{closed form}, in the sense that the conditional expectation functions $G_{\theta,t}:\R^{d\times t}\to\R^m$, given by
\begin{equation}\label{Gdef}
  G_{\theta,t}(x_1,\dots,x_t)= \E^\Q[\phi_{\theta}(x_1,\dots,x_t,X_{t+1},\dots,X_T)] ,
\end{equation}
can be efficiently evaluated at very low computational cost. As a result, we obtain the approximate value process
\begin{equation}\label{eq:rep_mart}
    V_{\theta,t} =  G_{\theta,t}(X_1,\dots,X_t)^\top\beta_\theta ,
\end{equation}
in closed form. This requirement is a key distinction of our method from other methods in the literature. We obtain the value process of the portfolio by regressing against its terminal value. Our approach, therefore, falls within the \say{regress later} category first mentioned in \cite{glasserman2002simulation}. As we will see in the numerical examples below, this approach performs better than the alternative \say{regress now}.

So, how good is our approximation \eqref{eq:rep_mart}? By Doob's inequality, making use of the martingale property of the value process, we obtain an upper bound on the pathwise maximal $L^1_\Pa$-error,
\begin{equation}\label{elembound}
 \textstyle \left\| \max_{t\le T} |V_t - V_{\theta,t}|\right\|_{L^1_\Pa} \le  \|\frac{d\Pa}{d\Q} \|_{L^2_\Q} \,\left\|\max_{t\le T} |V_t -V_{\theta,t}|\right\|_{L^2_\Q}  \le 2   \|\frac{d\Pa}{d\Q} \|_{L^2_\Q}  \| f-f_{\theta}\|_{L^2_\Q} .
\end{equation}
Note that $\|\frac{d\Pa}{d\Q} \|_{L^2_\Q}$ is known by the modeler. The $L^2_\Q$-approximation error on the right hand side of \eqref{elembound} is the objective in \eqref{eqnbetaAnew}. In practice, it can be estimated by Monte Carlo based on the training sample used to learn $f_\theta$. Hence, albeit elementary, this inequality gives a practical upper bound on the relevant $L^1_\Pa$-error of the approximation of $V_t$.

\subsection{Feature learning}
We now learn the parameter $\theta$ from the data. This is a second key distinction of our method, which allows to tackle the notorious curse of dimensionality that comes with polynomial feature maps, as we shall see below. Thereto we minimize the approximation error by an optimal choice of $\theta$, which leads to the non-convex optimization problem
\begin{equation}\label{optthebet}
      \min_{(\theta,\beta)\in\Theta\times\R^m} \| f -\phi_\theta^\top\beta\|_{L^2_\Q}.
 \end{equation}
Here is an elementary existence result.
\begin{lemma}\label{lemexi}
Assume
\begin{enumerate}
 \item $\Theta$ is compact,
 \item $\theta\mapsto \phi_{\theta,i}: \Theta\to L^2_\Q$ is continuous for all $i$,
 \item\label{lemexi3} $\{\phi_{\theta,1},\dots,\phi_{\theta,m}\}$ is linearly independent in $L^2_\Q$ for all $\theta\in\Theta$.
\end{enumerate}
Then there exists a solution to \eqref{optthebet}.
  \end{lemma}

The assumptions in Lemma \ref{lemexi} cannot be relaxed in general. This is shown by the following example.

\begin{example}
Let $\Q=\frac{1}{2}\sum_{i=1}^2 \delta_{x^{(i)}}$ be a discrete measure supported on two points $x^{(1)}\neq x^{(2)}$, so that every function $f\in L^2_\Q$ can be identified with the $\R^2$-vector $(f(x^{(1)}),f(x^{(2)}))$. We let $f=(1,0)$, $m=1$, and either
\begin{enumerate}

\item $\Theta = (0,1]$ (\emph{not compact}) and $\phi_\theta = \frac{(1,\theta)}{\sqrt{1+\theta^2}}$, or

\item $\Theta = [0,1]$ and $\phi_\theta =  \frac{(1_{(0,1]}(\theta),\theta+1_{\{0\}}(\theta))}{\sqrt{1+\theta^2}}$, so that $\phi_0=(0,1)$ (\emph{not continuous}), or

\item $\Theta = [0,1]$ and $\phi_\theta = \theta\frac{(1,\theta)}{\sqrt{1+\theta^2}}$, so that $\phi_0=(0,0)$ (\emph{not linearly independent}).

\end{enumerate}
For either case, we have $\inf_{(\theta,\beta)\in \Theta\times\R} \|f - \phi_\theta\beta\|_{\R^2} = \lim_{\theta\to 0} \|f - \phi_\theta\beta_\theta\|_{\R^2} =0$, but the infimum is not attained, $\|f - \phi_\theta\beta\|_{\R^2}>0$ for all $(\theta,\beta)\in \Theta\times\R$.
\end{example}

In practice, there are many factors that determine whether the approximation \eqref{optthebet} will be close to the true $f$. One of them is the relationship between the dimension of the random driver, $dT$, and the size of the training sample. Since in real-world applications the sample size is limited by practical constraints, it is necessary to reduce the effective dimension of the stochastic driver. This motivates the use of a linear dimensionality reduction, as follows. We henceforth assume that the feature map is parametrized in the form
\begin{equation}\label{phigAb}
 \phi_{\theta,i}(x) = g_i(A^\top x + b)
\end{equation}
for some exogenously given functions $g_i:\R^p\to \R$, $i=1,\dots,m$, for some $p\in\N$, and the parameter $\theta={ (A,b)}$ consists of a weight matrix $A$ and a bias vector $b$, for a subset $\Theta\subseteq { \R^{dT\times p}\times\R^p}$.

As a notational convention, to capture the evaluation of $\phi_{\theta}$ at the $\R^{d\times T}$-valued stochastic driver $X=(X_1,\dots,X_T)$, we decompose the $dT\times p$-matrix $A$ into $T$ consecutive blocks $A_t \in \R^{d\times p}$ such that $A^\top=(A_1^\top,\dots,A_T^\top)$. Then we have $\textstyle A^\top \vect(X) = \sum_{t=1}^T A_t^\top X_t$ and
\[  \textstyle\phi_{\theta,i}(X) \equiv \phi_{\theta,i}(\vect(X)) = g_i(\sum_{t=1}^T A_t^\top X_t+b).\]
The conditional expectation functions in \eqref{Gdef} read component-wise for $i=1,\dots,m$ as
\begin{equation}\label{eqnCExp}
 \textstyle G_{\theta,t,i}(x_1,\dots,x_t)=\E^\Q [g_i(b+\sum_{s=1}^t A_s^\top x_s +\sum_{s=t+1}^T A_s^\top X_s)].
\end{equation}

We remark that our approach does not replace the original stochastic driver $X$ by $A^\top X+b$ in general, only for the specific portfolio $f$ to which it is calibrated through \eqref{optthebet}. In the sequel, we study three different specifications of the type \eqref{phigAb} in more detail: a full polynomial basis, a polynomial feature map with linear dimensionality reduction, and a shallow neural network.

\subsection{Full polynomial basis} \label{sec:full_pol_base_part1}

We start with the non-weighted and non-biased case. We formally let $\Theta=\{(I_{dT},0)\}$ be the singleton consisting of the $dT\times dT$-identity matrix $A=I_{dT}$ and zero bias vector $b=0$. Accordingly, we omit the parameter and write shorthand $\phi_{(I_{dT},0)}\equiv \phi$. Problem \eqref{optthebet} boils down to the projection \eqref{eqnbetaAnew}. We let the feature map $\phi$ be composed of a basis of the space of all polynomials of degree $\delta$ or less,
\[ {\rm Pol}_\delta(\R^{dT})= \spn\{x^{\bm\alpha}\mid \bm\alpha\in\N_0^{dT},\,|\bm\alpha|\le \delta\}.  \]
In order that $\phi_i\in L^2_\Q$, we assume that
\begin{equation}\label{assdeltaint}
  \E_\Q[\|X\|^{2\delta}]<\infty.
\end{equation}

No attempt is made at selecting individual basis elements from within ${\rm Pol}_\delta(\R^{dT})$ at this stage, every polynomial is used in the projection. As a consequence---albeit leading to closed form expressions in \eqref{GfullPB} below---this feature map suffers the curse of dimensionality from the rapid growth of the number of basis functions $m$ as a function of $d$ and $T$,
\[\textstyle m=\dim {\rm Pol}_\delta(\R^{dT}) = \binom{dT + \delta}{dT}. \]
Table \ref{tabdimPol} shows that the dimension $m$ quickly becomes larger than the training sample size in practice.\footnote{Strictly speaking, the dimension of the linear span of the functions $\phi_i$ in $L^2_\Q$ could be less than $m$, because they may be linearly dependent as elements in $L^2_\Q$. This is in particular the case when $\Q$ is an empirical measure from a sample of size $n$, as described in Subsection \ref{sec:fin_sample_approx}. In this case, $m>n$ would simply lead to an exact interpolation of $f$, which likely will result in overfitting.}

\begin{table}[ht]
\centering
  \begin{tabular}{r||r |r}
$T$   &     \multicolumn{1}{c|}{$d=3$ }  &      \multicolumn{1}{c}{$d=5$}  \\
\hline
5 & 816&   3{,}276   \\
40 & 302{,}621&   1{,}373{,}701
\end{tabular}
\caption{Dimension $m=\dim {\rm Pol}_\delta(\R^{dT})$ for $\delta=3$.}\label{tabdimPol}
\end{table}

The exact form of the conditional expectation \eqref{Gdef} depends on the choice of $\phi$ and the distribution $\Q=\otimes_{t=1}^T \Q_t$ of $X$. Choosing an orthogonal basis of polynomials in $L_{\Q}^2=\otimes_{t=1}^T L^2_{\Q_t}$ can greatly simplify the calculations. In view of \cite[Theorem 8.25]{sul_15}, there is a system of orthogonal polynomials $\{g_{\bm\alpha}\mid \bm{\alpha} \in \N_0^{dT}\}$ on $\R^{dT}$ for $\Q$ that can be expressed as
\[ g_{\bm\alpha}(x) = \prod_{t=1}^{T} h_{t,\bm\alpha_t}(x_t),\quad x=\vect(x_1,\dots,x_t), \quad \bm{\alpha}=(\bm\alpha_1,\dots,\bm\alpha_T), \]
where $\{h_{t,\bm\alpha_t}\mid \bm{\alpha_t}\in \N_0^{d}\}$ is a system of orthogonal polynomials on $\R^d$ for $\Q_t$, for every $t=1,\dots,T$. Moreover, $\deg g_{\bm \alpha}=|\bm\alpha|$ and $\deg h_{t,\bm \alpha_t}=|\bm\alpha_t|$. Now choose an index mapping \footnote{This index mapping only performs the ordering of the elements which is required to conveniently and formally write the conditional expectation function. This index mapping does not perform any selection of a subset of the basis and any mapping would yield the same final results.} $\{1,\dots,m\}\ni i \mapsto \bm \alpha^i \in \{\bm\alpha\in \N_0^{dT}\mid |\bm\alpha|\le\delta\}$, then we obtain an orthogonal basis of ${\rm Pol}_\delta(\R^{dT})$ for $\Q$ by setting $\phi_i = g_{\bm\alpha^i}$. The conditional expectation functions \eqref{eqnCExp}, where we omit $\theta$, are then by orthogonality of $h_{s,\bm \alpha^i_s}$ given in closed form as
\begin{equation}\label{GfullPB}
  G_{t,i}(x_1,\dots,x_t) = \prod_{s=1}^{t} h_{s,\bm\alpha^i_s}(x_s)\prod_{s=t+1}^T 1_{\bm \alpha^i_{s} = \bm 0} .
\end{equation}
This extends to the unconditional expectations, $\E^\Q[\phi_i(X)]=\phi_i(0)1_{\bm \alpha^i=\bm 0}$.

\begin{example}
Consider the multinormal case $\vect(X)\sim N(0,I_{d T })$. Here we can choose $g_{\bm\alpha}$, and $h_{t,\bm\alpha}\equiv h_{\bm\alpha}$, as the multivariate (probabilists') Hermite polynomials of order $\bm{\alpha}\in\N_0^{dT}$ on $\R^{dT}$, and order $\bm{\alpha}\in\N_0^{d}$ on $\R^{d}$, respectively.
\end{example}

\subsection{Polynomial feature map with linear dimensionality reduction} \label{sec:ldr_part1}

We now tackle the curse of dimensionality of the above full polynomial basis. Thereto we let $p\le dT$ and ${ \{g_1,\dots,g_m\}}$ be a basis of ${ {\rm Pol}_\delta(\R^{p})}$, and we consider all feature maps \eqref{phigAb} for weight matrices $A$ with full rank and bias vectors $b$. As above, we assume that \eqref{assdeltaint} holds, so that $\phi_{(A,b),i}\in L^2_\Q$. The following theorem shows that we can assume that $b=0$ and $A$ lies in the \say{Stiefel manifold} $V_p(\R^{dT}) =\{ A\in\R^{ dT \times p}\mid A^\top A={I}_p\}$, the set of all orthonormal $p$-frames in $\mathbb{R} ^{dT}$.

\begin{theorem}\label{thmVk}
For any $A , \tilde A\in\R^{dT\times p}$ with full rank and $b\in\R^p$, the following are equivalent:
\begin{enumerate}
  \item $\tilde A\in V_p(\R^{dT})$ and $\spn\{\phi_{(A,b),1},\dots,\phi_{(A,b),m}\}=\spn\{\phi_{(\tilde A,0),1},\dots,\phi_{(\tilde A,0),m}\}$,
  \item $\tilde A=A(A^\top A)^{-1/2}U$ for some orthogonal $p\times p$-matrix $U$.
\end{enumerate}
\end{theorem}

In view of Theorem \ref{thmVk} the parameter set can be chosen to be $\Theta=V_p(\R^{dT})$ yielding feature maps of the form \[\phi_{(A,0),i}(x)\equiv\phi_{A,i}(x)=g_i(A^\top x)\]
without loss of generality. We arrive at the following existence and non-uniqueness result.
\begin{theorem}\label{thmpolyuniqueN}
For the polynomial feature map, there exists a minimizer in $\Theta= V_p(\R^{dT})$ of \eqref{optthebet}. However, uniqueness does not hold, in the sense that the optimal subspace $\spn\{\phi_{A,1},\dots,\phi_{A,m}\}$ is not unique, in general.
\end{theorem}

Problem \eqref{optthebet} corresponds to a \say{linear dimensionality reduction} with matrix manifold $V_p(\R^{dT})\times \R^m$ in the spirit of \cite{cunningham2015linear}. The dimensionality reduction is produced exclusively by the linear mapping $A$ of $\R^{dT}$ onto $\R^p$. There is no other restriction imposed on ${\rm Pol}_\delta(\R^{p})$, every polynomial basis function $g_i$ is used. And yet, the dimensionality reduction compared to the full polynomial basis of ${\rm Pol}_\delta(\R^{dT})$ is significant. Indeed, the total dimension of the optimization problem \eqref{optthebet} is given by the sum of $\dim V_p(\R^{dT}) = dT p -\frac{1}{2}p(p+1)$ plus $m=\dim{ {\rm Pol}_\delta(\R^{p})}$. This sum can be kept low by choosing $p$ small enough. Table \ref{tabdimPolLDR} shows that the total dimension of \eqref{optthebet} remains moderate compared to the corresponding figures of the full polynomial basis from Table \ref{tabdimPol}.

\begin{table}[ht]
\centering
  \begin{tabular}{r||r |r}
$T$   &    \multicolumn{1}{c|}{$d=3$, $p=3$}  &      \multicolumn{1}{c}{$d=5$, $p=10$} \\
\hline
5 & $9+20=29$ &   $195 + 286=481$   \\
40 & $114+20=134$ &   $1{,}945 + 286= 2{,}231$
\end{tabular}
\caption{Total dimension $\dim V_p(\R^{dT})+\dim {\rm Pol}_\delta(\R^{p})$ for $\delta=3$.}\label{tabdimPolLDR}
\end{table}

The calculation of the conditional expectation \eqref{Gdef} is not as simple as for the full polynomial basis in \eqref{GfullPB}. Instead we need to compute the unconditional moments \eqref{eqnCExp}, which here reduce to
\begin{equation}\label{eqnCExpLDR}
 \textstyle G_{A,t,i}(x_1,\dots,x_t)=\E^\Q [g_i( \sum_{s=1}^t A_s^\top x_s +\sum_{s=t+1}^T A_s^\top X_s)],
\end{equation}
for $A\in V_p(\R^{dT})$. Evaluation of \eqref{eqnCExpLDR} boils down to compute multivariate moments of the $\R^p$-valued random variable $Y= \sum_{s=1}^t A_s^\top x_s +\sum_{s=t+1}^T A_s^\top X_s$. Thereto we utilize \cite[Lemma 1]{kan_08}, which relates multivariate moments to univariate moments, generalizing $4y_1 y_2=(y_1+y_2)^2 - (y_1-y_2)^2$, by
\begin{equation}\label{kaneq}
 \textstyle   y^{\bm \alpha}\equiv y_1^{\alpha_1}\cdots y_p^{\alpha_p} =  \frac{1}{|\bm\alpha| !}  \sum_{\bm\nu =\bm 0}^{\bm \alpha}  (-1)^{|\bm \nu|}  {\alpha_1\choose \nu_1}\cdots {\alpha_p\choose \nu_p}  ( h_{\bm\alpha,\bm\nu}^\top  y)^{|\bm \alpha|}   ,
\end{equation}
for vectors $h_{\bm\alpha,\bm\nu}=(\alpha_1/2-\nu_1,\dots,\alpha_p/2-\nu_p)^\top$. The sum in \eqref{kaneq} has in effect $(\alpha_1+1)\cdots(\alpha_p+1)/2$ terms. For $\delta=3$, this amounts to maximal $2^3/2=4$ terms.  As a result, the evaluation of \eqref{eqnCExpLDR}, that is, the computation of $\E^\Q [Y^{\bm\alpha}]$, reduces to the calculation of the $|\bm\alpha|$th moments of the scalar random variables $h_{\bm\alpha,\bm\nu}^\top Y$ in \eqref{kaneq}, which are given in closed form for various distributions of $X$.

\begin{example}
For the multinormal case $\vect(X)\sim N(0,I_{dT})$ we have \[\textstyle h_{\bm\alpha,\bm\nu}^\top Y\sim N(\sum_{s=1}^t h_{\bm\alpha,\bm\nu}^\top A_s^\top x_s  , \sum_{s=t+1}^T  \|A_s h_{\bm\alpha,\bm\nu}\|^2).\] These univariate moments are given in closed form, as explicitly stated in \cite[Proposition 2]{kan_08}.

\end{example}

\subsection{Shallow neural network}\label{sec:nn_part1}
In this third specification, we consider a shallow neural network with the rectified linear unit (\say{ReLU}) activation function. More specifically, we let $p=m$ and $g_i(y)=y_i^+$, for $i=1,\dots,m$. This yields the feature maps \eqref{phigAb} of the form
\[ \phi_{(A,b),i}(x)\equiv \phi_{(a_i,b_i)}(x) = (a_i^\top x + b_i)^+ \]
for weight matrices $A=(a_1,\dots,a_m)\in\R^{dT\times m}$ and bias vectors $b\in\R^m$. Henceforth we assume that \eqref{assdeltaint} holds for $\delta=1$, so that $\phi_{(a_i,b_i)}\in L^2_\Q$. By the positive homogeneity of the components of the feature map in the parameter, $\phi_{(\lambda a_i,\lambda b_i)} = \lambda\phi_{(a_i,b_i)}$ for all $\lambda> 0$, we can assume that $ (a_i,b_i)$ lies in the \say{unit sphere} $\Scal_{dT}$ in $\R^{dT+1}$ . Hence the parameter set can be chosen as the compact product manifold $\Theta = (\Scal_{dT})^m$ without loss of generality.

What about linear independence of $\phi_{(a_i,b_i)}$? Here is a fundamental result, which seems to be little known in the literature.

\begin{theorem}\label{thmReLUspanNew}
For any $(a_i,b_i),(\tilde a_i,\tilde b_i)\in \Scal_{dT}$, $i=1,\dots,m$, the following statements hold:
\begin{enumerate}
  \item\label{thmReLUspanNew1} If
  \begin{equation}\label{thmReLUspanNew1eq}
   \spn\{ \phi_{(a_1,b_1)},\dots,\phi_{(a_m,b_m)}\}=\spn\{ \phi_{(\tilde a_1,\tilde b_1)},\dots,\phi_{(\tilde a_m,\tilde b_m)}\}
  \end{equation}
  then $\{\pm (a_1,b_1),\dots, \pm (a_m,b_m)\}=\{\pm (\tilde a_1,\tilde b_1),\dots, \pm (\tilde a_m,\tilde b_m)\}$.

\item\label{thmReLUspanNew2} If $(a_i,b_i)\neq \pm (a_j,b_j)$ for all $i\neq j$ then
\begin{equation}\label{thmReLUspanNew2eq}
  \text{$\{ \phi_{(a_1,b_1)},\dots,\phi_{(a_m,b_m)}\}$ is linearly independent.}
\end{equation}

\end{enumerate}

\end{theorem}

Note that the converse implication in Theorem \ref{thmReLUspanNew}\ref{thmReLUspanNew1} is not true, as can easily be seen from the case where $m=1$ and $(\tilde a_1,\tilde b_1)=-(a_1,b_1)\in \Scal_{dT}$ with $a_1\neq 0$. What's more, the following example shows that the assumptions in Theorem \ref{thmReLUspanNew}\ref{thmReLUspanNew2} cannot be relaxed to pairwise inequality, $(a_i,b_i)\neq (a_j,b_j)$ for all $i\neq j$.

\begin{example}\label{exlinind}
Let $(a_1,b_1),\dots,(a_3,b_3)\in \Scal_{dT}$ be linearly dependent vectors such that $\sum_{i=1}^3 c_i (a_i,b_i)=0$ for some coefficients $c_i\neq 0$. Define $(a_{3+i},b_{3+i})=-(a_i,b_i)\in \Scal_{dT}$, $i=1,2,3$. Then $\{\phi_{(a_i,b_i)}\mid i=1,\dots,6\}$ is linearly dependent,
\[\textstyle\sum_{i=1}^3 c_i  \phi_{(a_i,b_i)} +\sum_{i=1}^3 (-c_i) \phi_{(a_{3+i},b_{3+i})}    = \sum_{i=1}^3 c_i (a_i^\top x+b_i) = 0,\]
while $(a_i,b_i)\neq (a_j,b_j)$ for all $i\neq j$.
\end{example}

We conclude that the assumptions of the existence Lemma \ref{lemexi} are not met. Indeed, we have the following non-existence result, which contrasts somewhat surprisingly with the widespread use of ReLU neural networks in machine learning.

\begin{theorem}\label{thmNEReLU}
For the shallow ReLU neural network, there exists no minimizer of \eqref{optthebet} in general. Moreover, uniqueness does not hold, in the sense that the optimal subspace $\spn\{\phi_{(a_1,b_1)},\dots,\phi_{(a_m,b_m)}\}$ is not unique, in general.
\end{theorem}

\begin{remark}\label{remNN}
The proof of the non-existence statement in Theorem \ref{thmNEReLU} is by means of a counterexample. It shows that the space of functions represented by the finite shallow neural network $\{ \phi_\theta^\top\beta\mid (\theta,\beta)\in\Theta\times \R^m\}$ is not closed in $L^2_\Q$ in general. This result also holds for finite neural networks with several layers, as recently shown in \cite{pet_ras_voi_21}. Note that this finding is consistent with the celebrated universal approximation property of neural networks \cite{hor_sti_whi_89}, which holds asymptotically for $m\to\infty$. Indeed, in our case, $m$ is fixed and finite. In view of their spectacular performance in solving practical problems, neural networks have become the subject of intensive research. We refer to the literature overview given in \cite{pet_ras_voi_21}. An important strand of research is focused on the training of neural networks. Despite the non-convexity of the objective function \eqref{optthebet} in $(\theta,\beta)$, researchers have shown that its empirical landscape does not exhibit non-global local minima with high probability if the problem is sufficiently overparametrized. That is, if $m$ is sufficiently large compared to the sample size, see \cite{ven_ban_bru_19}. Also we exploit these good empirical properties of neural networks in the numerical case studies below. Another line of research focuses directly on financial payoffs. \cite{bos_car_pap_21} show that an infinitely large shallow neural network is capable of uniquely replicating any payoff function, and they find the representation in terms of the spectral decomposition of a Volterra integral equation.
\end{remark}

The conditional expectation functions \eqref{eqnCExp} read here as
\begin{equation}\label{eqnCExpReLU}
 \textstyle G_{(a_i,b_i),t}(x_1,\dots,x_t)=\E^\Q [( b_i+ \sum_{s=1}^t a_{i,s}^\top x_s +\sum_{s=t+1}^T a_{i,s}^\top X_s)^+],
\end{equation}
where we decompose every column vector $a_i$ of $A$ into $T$ consecutive blocks $a_{i,t}$ such that $a_i^\top = (a_{i,1}^\top,\dots,a_{i,T}^\top)$.  Evaluation of \eqref{eqnCExpReLU} boils down to compute $\E^\Q [Y^+]$ for the scalar random variable $Y=b_i+ \sum_{s=1}^t a_{i,s}^\top x_s +\sum_{s=t+1}^T a_{i,s}^\top X_s$, which is given in closed form for various distributions of $X$.

\begin{example}
 For the multinormal case $\vect(X)\sim N(0,I_{dT})$, we have
\begin{equation}\label{YReLU}
 \textstyle Y\sim N(b_i+ \sum_{s=1}^t a_{i,s}^\top x_s , \sum_{s=t+1}^T \| a_{i,s}\|^2).
\end{equation}
We then obtain a closed form expression for \eqref{eqnCExpReLU} by combining \eqref{YReLU} with the well known Bachelier's call option price formula $\E^\Q[Z^+] = \mu \Phi(\mu/\sigma) + \sigma \Phi'(\mu/\sigma)$, for a normal distributed random variable $Z\sim N(\mu,\sigma^2)$, where $\Phi$ denotes the standard normal distribution function and $\Phi'$ its density function, see, e.g., \cite[Section 4.3]{del_sch_06} or \cite{fernandez2020}.
\end{example}

In other cases, where the extended Fourier transforms $\widehat\Q_t(u) = \E_\Q[\e^{u^\top X_t}]$ of the marginal distributions $\Q_t$ of $X_t$ are given in closed form, for a suitable domain of complex-vector valued arguments $u$, we can utilize Fourier transform analysis. Indeed, for any constant $w>0$, we have the identity
\[ \textstyle y^+ = \frac{1}{2\pi}\int_\R \e^{(w+\im\lambda)y} \frac{1}{(w+\im\lambda)^2} \,d\lambda. \]
Hence the evaluation of \eqref{eqnCExpReLU} reduces to the computation of the line integral
\begin{equation}\label{eqnCExpReLUFourier}
  \textstyle G_{(a_i,b_i),t}(x_1,\dots,x_t)=  \frac{1}{2\pi}\int_\R \widehat F_Y(w+\im\lambda) \frac{1}{(w+\im\lambda)^2} \, d\lambda ,
\end{equation}
where $\widehat F_Y(w+\im\lambda)= \E_\Q[\e^{(w+\im\lambda) Y}] =  \e^{(w+\im\lambda)( b_i+ \sum_{s=1}^t a_{i,s}^\top x_s)} \prod_{s=t+1}^T \widehat\Q_s((w+\im\lambda) a_{i,s})$ is in closed form. Note that Fourier type integrals like the one in \eqref{eqnCExpReLUFourier} are routinely computed in finance applications, e.g, in L\'evy type or affine models, \cite{duf_fil_sch_03}. So one can draw on existing libraries of computer code.


\subsection{Finite-sample estimation} \label{sec:fin_sample_approx}

While surprising and remarkable, the non-uniqueness and non-existence results in Theorems \ref{thmpolyuniqueN} and \ref{thmNEReLU} for polynomial feature maps with dimensionality reduction and shallow ReLU neural networks, respectively, are mainly of theoretical interest, arguably. See also Remark~\ref{remNN}. In practice, we solve \eqref{optthebet} numerically using some quasi-Newton algorithm, which finds local minima that serve as approximate solutions. Thereto, we replace the model population measure $\Q$ by the empirical measure $\widehat\Q=\frac{1}{n}\sum_{i=1}^n \delta_{x^{(i)}}$ based on a training sample $ x^{(1)},\dots,x^{(n)}$ drawn from $\Q$, along with the corresponding function values $y_i=f(x^{(i)})$.

For the full polynomial basis, problem \eqref{optthebet} boils down to the projection \eqref{eqnbetaAnew}, and we obtain the optimal
\begin{equation}\label{hatbetaPhi}
 \widehat\beta =  \textstyle(\frac{1}{n}\Phi^\top\Phi)^{-1} (\frac{1}{n} \Phi^\top y)
\end{equation}
where we define $\Phi\in\R^{n\times m}$ by $\Phi_{ij}=\phi_{j}(x^{(i)})$. This empirical estimator is consistent. The law of large numbers implies that $\widehat\beta$ converges in probability to the optimal $\beta$ in \eqref{eqnbetaAnew} for the model population measure $\Q$, as the sample size $n\to\infty$. Moreover, the central limit theorem holds and theoretical guarantees for the sample error can be established, see, e.g., \cite{bou_fil_21}.

For the polynomial feature map with dimensionality reduction, we use the Riemannian BFGS algorithm \cite{huang2015broyden} to find a local minimizer of \eqref{optthebet} over the Riemannian manifold $V_p(\R^{dT})\times\R^m$. 

For the shallow ReLU neural network, we use the BFGS algorithm in the {\it Scikit-learn} library \cite{scikit-learn} for the Python programming language to find a local minimizer of \eqref{optthebet} over the full parameter set $\R^{dT\times m}\times\R^m\times \R^m$. 

Given the lack of uniqueness, and even existence, for the polynomial feature map with dimensionality reduction, and the shallow ReLU neural network, it remains an open research question whether asymptotic consistency holds and theoretical guarantees can be established for these specifications.

\section{European call option example revisited} \label{sec:call_example}

Having presented the theoretical background, we now turn back to our illustrating preliminary example in Section \ref{sec:call_nmc_new}. We apply the functional bases described in Sections \ref{sec:full_pol_base_part1}--\ref{sec:nn_part1}, following the steps outlined in Section \ref{sec:fin_sample_approx}. For each functional basis we show the same quality metrics as in Section~\ref{sec:call_nmc_new} in order to compare to the results from nested Monte Carlo estimation.

Additionally, each functional basis is also compared to other related methods, such as regress-now LSMC. When comparing between regression-based methods, we use an additional quality metric: the mean $L_1$ error over the empirical distribution of $\Delta V_1$. This metric allows us to make comparisons of the goodness-of-fit along the entire distribution, not only the tails. For more details about the quality metrics use, we refer to Appendix~\ref{sec:qual_metrics}.

\subsection{Results}

The first comparison uses the full polynomial basis described in Section~\ref{sec:full_pol_base_part1}. In Tables \ref{tab:call_mae_pv_all_methods} and \ref{tab:call_mae_es_all_methods} we present the MApE comparison among nested Monte Carlo (nMC), regress-now polynomial basis and the replicating martingale (regress-later) full polynomial basis. We can see how the replicating martingale outperforms the other two methods in the estimation of the present value and the 99\% expected shortfall. For a more comprehensive comparison, we look at the mean $L_1$ error in Table~\ref{tab:call_l1_error_all_methods}. We can see that Table~\ref{tab:call_mae_es_all_methods} confirms the conclusions from Table~\ref{tab:call_l1_error_all_methods}, namely that the replicating martingale estimators outperform the regress-now estimators. In this regard, we verify what others in the literature have reported before for regress-later estimators.

\begin{table}[ht]
    \centering
    \small
    \caption{European call, comparison of present value MApE (in percentage points)}

\begin{tabular}{@{}rrrrrrrr@{}}
\toprule
\multicolumn{2}{r}{\textbf{}} & \multicolumn{2}{r}{\textbf{Full Polynomial basis}} &  \textbf{LDR} & \multicolumn{2}{r}{\textbf{Neural Network}} \\
\textbf{Samples} & \textbf{nMC} &                    Regress-now & Regress-later & Regress-later &             Regress-now & Regress-later \\
\midrule
\multicolumn{1}{l}{\textbf{\textbf{\textbf{T}}=5}}  &   &   &   &   &   &   \\
           1,000 &          6.5 &                            4.0 &           1.3 &           0.5 &                     4.2 &           0.2 \\
           5,000 &          4.0 &                            1.8 &           0.2 &           0.2 &                     1.8 &           0.1 \\
          10,000 &          3.8 &                            1.2 &           0.1 &           0.1 &                     1.2 &           0.1 \\
          50,000 &          1.7 &                            0.6 &           0.1 &           0.1 &                     0.6 &           \textless{}0.1 \\
\multicolumn{1}{l}{\textbf{\textbf{\textbf{T}}=40}}  &   &   &   &   &   &   \\
           1,000 &          7.5 &                            7.2 &             &           3.8 &                     7.9 &           5.5 \\
           5,000 &          3.9 &                            3.2 &             &           2.1 &                     3.2 &           1.7 \\
          10,000 &          3.1 &                            2.3 &             &           1.4 &                     2.4 &           0.9 \\
          50,000 &          2.4 &                            1.0 &             &           0.5 &                     1.0 &           0.2 \\
\bottomrule
\end{tabular}

    \label{tab:call_mae_pv_all_methods}
\end{table}

\begin{table}[ht]
    \centering
    \small
    \caption{European call, comparison of expected shortfall MApE (in percentage points)}

\begin{tabular}{@{}rrrrrrrr@{}}
\toprule
\multicolumn{2}{r}{\textbf{}} & \multicolumn{2}{r}{\textbf{Full Polynomial basis}} &  \textbf{LDR} & \multicolumn{2}{r}{\textbf{Neural Network}} \\
\textbf{Samples} & \textbf{nMC} &                    Regress-now & Regress-later & Regress-later &             Regress-now & Regress-later \\
\midrule
\multicolumn{1}{l}{\textbf{\textbf{\textbf{T}}=5}}  &   &   &   &   &   &   \\
           1,000 &         26.7 &                           21.5 &           4.8 &           2.0 &                    47.5 &           0.9 \\
           5,000 &         15.9 &                            9.3 &           1.1 &           0.9 &                    11.5 &           0.2 \\
          10,000 &         14.5 &                            6.5 &           0.8 &           0.7 &                     7.3 &           0.1 \\
          50,000 &          7.9 &                            2.9 &           0.6 &           0.5 &                     3.6 &          \textless{}0.1 \\
\multicolumn{1}{l}{\textbf{\textbf{\textbf{T}}=40}}  &   &   &   &   &   &   \\
           1,000 &        403.6 &                          141.0 &             &          10.0 &                   459.0 &          16.0 \\
           5,000 &        110.9 &                           46.8 &             &           6.5 &                   106.9 &           4.9 \\
          10,000 &         56.6 &                           29.6 &             &           4.9 &                    57.3 &           3.3 \\
          50,000 &         19.1 &                           12.5 &             &           2.9 &                    11.4 &           1.0 \\
\bottomrule
\end{tabular}

    \label{tab:call_mae_es_all_methods}
\end{table}

\begin{table}[ht]
    \centering
    \small
    \caption{European call, comparison of relative mean $L_1$ error (in percentage points)}

\begin{tabular}{@{}rrrrrrrr@{}}
\toprule
       \textbf{} & \multicolumn{2}{r}{\textbf{Full Polynomial basis}} &  \textbf{LDR} & \multicolumn{2}{r}{\textbf{Neural Network}} \\
\textbf{Samples} &                    Regress-now & Regress-later & Regress-later &             Regress-now & Regress-later \\
\midrule
\multicolumn{1}{l}{\textbf{\textbf{\textbf{T}}=5}}  &   &   &   &   &   \\
           1,000 &                           15.5 &           4.6 &           1.1 &                    36.6 &           0.5 \\
           5,000 &                            6.8 &           0.9 &           0.6 &                    12.5 &           0.4 \\
          10,000 &                            4.8 &           0.7 &           0.6 &                     8.0 &           0.4 \\
          50,000 &                            2.2 &           0.5 &           0.5 &                     2.9 &           0.4 \\
\multicolumn{1}{l}{\textbf{\textbf{\textbf{T}}=40}}  &   &   &   &   &   \\
           1,000 &                           26.5 &             &           4.6 &                    68.8 &           7.4 \\
           5,000 &                           11.9 &             &           2.6 &                    24.0 &           2.2 \\
          10,000 &                            8.5 &             &           2.0 &                    14.9 &           1.4 \\
          50,000 &                            3.9 &             &           1.3 &                     4.1 &           0.8 \\
\bottomrule
\end{tabular}

    \label{tab:call_l1_error_all_methods}
\end{table}

Tables \ref{tab:call_mae_pv_all_methods} and \ref{tab:call_mae_es_all_methods} do not show results for the full polynomial basis under the replicating martingale approach for $T=40$. The reason for this is the combinatorial explosion in the number of basis functions as the dimensionality of the problem grows. As shown in Table~\ref{tabdimPol}, the number of basis functions for $d=3, \delta=3, T=40$ is $302,621$. The number of samples would have to be at least of that magnitude, yielding a problem that, while feasible for some algorithms, is not necessarily practical in the real world. Our focus is to describe a method that shows good quality even for high-dimensional problem with a manageable number of samples.
It is important to note that the regress-now approach does not suffer from this problem and shows a quality improvement over the Monte Carlo approach.

The problems with high dimensional cases described in the previous paragraph motivate our use of linear dimensionality reduction (LDR), as described in Section~\ref{sec:ldr_part1}. Table \ref{tabdimPolLDR} shows that the number of parameters to be estimated can be greatly reduced, from 816 to 29 for $T=5$ and from 302,621 to 134 for $T=40$. Whether a function $f$ can be well approximated by a polynomial basis with linear dimensionality reduction for a small $\delta$ and $p$ depends on the nature of the function. Asymptotically, any function in $L^2_{\Q}$ can be approximated with arbitrary precision.

The optimization problem in \eqref{optthebet} is solved over the product manifold $V_k(\R^{dT}) \times \R^m$ rather over $\R^{dT \times k } \times \R^m$, which is supported by Theorem \ref{thmVk}. This reduces the effective dimensionality of the problem and simplifies the calculation of the conditional expectation of $\phi(X)$ in Equation~\eqref{eq:rep_mart}. For these examples we have used the Riemannian BFGS algorithm from \cite{huang2015broyden} using the C++ library published by the authors. Additionally, we tested another two algorithms, Riemannian Trust Regions (\cite{absil2007trust}) as implemented by the Python library \say{pymanopt} (\cite{townsend2016pymanopt}), and Grassmann Gauss-Newton (\cite{hokanson2018data}) as implemented by the authors in the publicly available Python library. In all cases, Riemannian BFGS achieved better results.

For this example, we have chosen $p=3$. We performed a sensitivity analysis on this parameter and found that a larger value might lead to better results in some cases but not in all cases. The full results are included in Appendix~\ref{subsec:sensitivity_ldr}. In summary, $p$ is an important hyper-parameter but it is not special in any way, which means that any robust method for selecting hyper-parameters---cross-validation, pilot simulation, sensitivity analysis---can be used. In Appendix~\ref{subsec:sensitivity_ldr}, we also provide a sensitivity analysis to the starting point of the optimization. Since BFGS is quasi-Newton method, it is not guaranteed to find a global minimum in a general case. In the case of the polynomial basis with LDR, we find that the selection of the starting point makes a big difference in the final result.

Tables \ref{tab:call_mae_pv_all_methods} and \ref{tab:call_mae_es_all_methods} show the results of applying LDR ($p=3$) to the European call problem and we can see how this method performs in relation to the other alternatives. We can see that the polynomial LDR has lower error than nested Monte Carlo and both regress-now and regress-later polynomials. It also becomes clear that the LDR approach allows high dimensional problems where the regress-later approach on a full polynomial basis would fail due to producing a very large number of basis functions, which leads to a computational problem due to time or memory constraints.

We next describe the results obtained using a neural network model as defined in Section~\ref{sec:nn_part1}. For this example we have chosen to work with $m=101$ nodes, whereof one is a bias term only, say $A_{101}=0$. The total number of parameters is thus $1{,}600+101=1{,}701$ for $T=5$, and $12{,}100+101=12{,}201$ for $T=40$. This choice was made via cross-validation.

The neural network was optimized via backpropagation using the BFGS algorithm from Python's popular library scikit-learn (\cite{scikit-learn}).  The results in Tables \ref{tab:call_mae_pv_all_methods} and \ref{tab:call_mae_es_all_methods} show an excellent quality of the neural network replicating martingale in this example, outperforming every other choice, except for the risk calculations with a very low number of samples (1,000). Also here we can see that the replicating martingale method (regress-later) outperforms the regress-now variation, the same way that it did for the polynomials.

In Figure~\ref{fig:box_plot_call_es} we can compare the empirical distribution of the errors for each method. This figure makes it easy to qualitatively assess the differences between the different methods, for example: the high variance of nested Monte Carlo vis-a-vis the lower variance of replicating martingales, or the higher accuracy of regress-later methods compared to regress-now methods. We can also see that, despite the non-linear optimization with random starting points involved, the neural network replicating martingale does not have qualitatively higher variance than the polynomial equivalents.

\begin{figure}[ht]
\includegraphics[width=\textwidth]{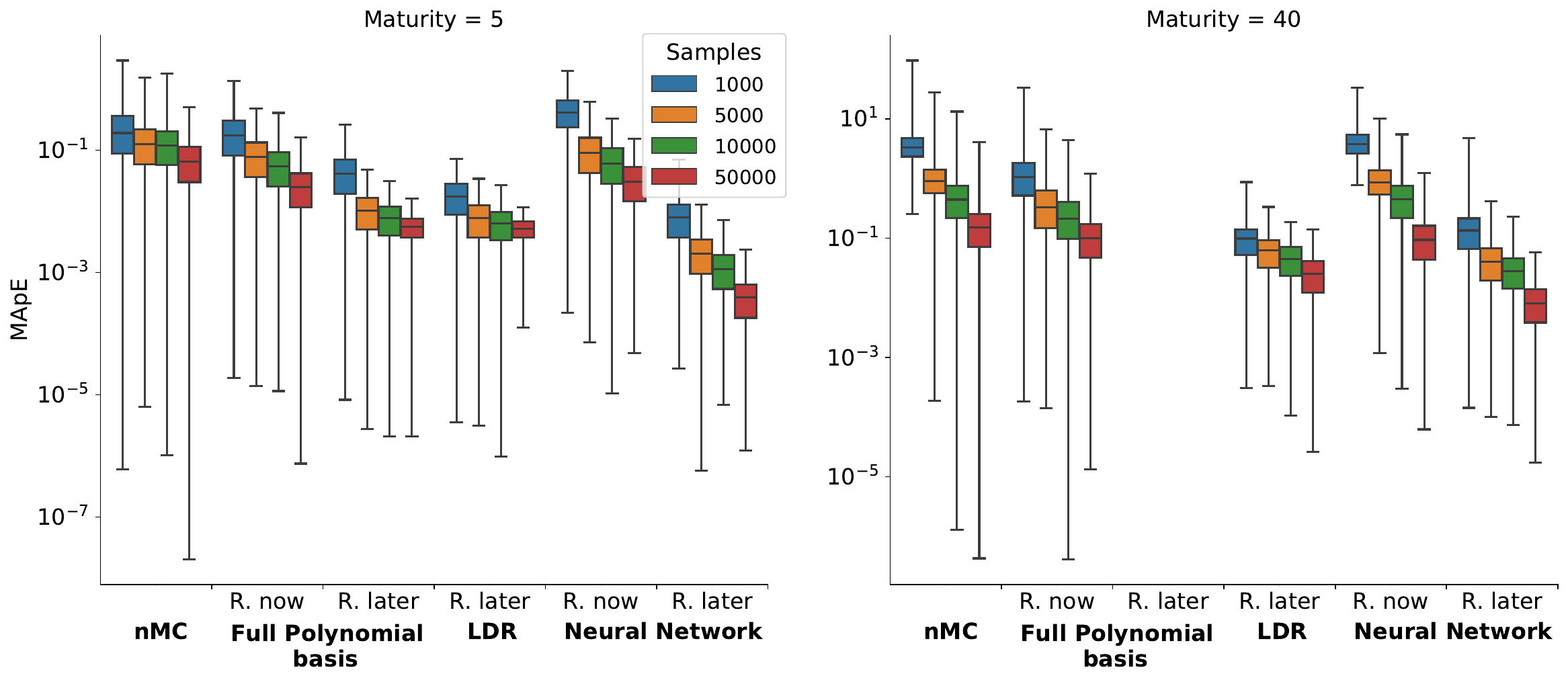}
\caption{Distribution of expected shortfall estimates per method for the European call example. Boxes show the upper and lower quartiles of the empirical distributions, while whiskers show their maxima and minima. Due to the logarithmic scale, a visual artifact is introduced by which the inter-quartile range seems to increase with larger sample sizes---most notably in the nested Monte Carlo (nMC) for Maturity=40. In reality, the inter-quartile range decreases with larger sample sizes, but the logarithmic scale makes the box bigger for smaller errors.}
\label{fig:box_plot_call_es}
\end{figure}

It is interesting to consider the structure of the neural network and polynomial models, to understand what they have in common and what they do not. As seen in \eqref{phigAb}, both methods use a linear map to reduce the dimensionality of the input before applying a non-linear function. The polynomial model is based on global polynomials while the neural network can be seen as a data-driven piece-wise linear model. While usually piece-wise linear models require a grid to be defined a priori, neural networks adjust the bias term to \say{place} the grid where it is most needed according to the input data.

A comparison of the runtimes is given in Appendix~\ref{app_runtime}.

\section{Insurance liability model example}\label{sec:insurance}

Having shown the effectiveness of learning the replicating martingale in the case of a European call option, we present now a more complex example: a variable annuity guarantee. Unlike the previous example, this one features path dependent cash flows at multiple points in time and also a dependency on a stochastic mortality model, rather than only stochastic market variables. The model has been built using models commonly in use in the insurance industry. The policyholder population is fictitious. We first describe the model. Then we present the numerical results following the same structure as for the European call option example in Section~\ref{sec:call_example}.

\subsection{Model}

The insurance product being simulated is an investment account with a \say{return premium on death} guarantee. Every policyholder has an investment account. At each time period, the policyholders pay a premium, which is used to buy assets. These assets are deposited in the fund, divided into the different assets according to a fix asset allocation, which is the same as the initial one. The value of the fund is driven by the inflows from premiums and the market value changes, which are driven by the interest rate, equity, and real estate models. At each time step, a number of policyholders die---as determined by the stochastic life table---and the investment fund is paid out to the beneficiaries. If the investment fund were below the guaranteed amount, the company will additionally pay the difference between the fund value and the guaranteed amount. The guaranteed amount is the simple sum of all the premiums paid over the life of the policy. Over the course of the simulation the premiums paid gradually increase the guaranteed amount for each policy.

All policies have the same maturity date. In the short-term, low dimensional example, the maturity is $T=5$. In the long term, high dimensional example the maturity is $T=40$. At maturity, the higher of the fund value and the guaranteed amount is paid out to all survivors.

The investment portfolio holds four assets: a ten-year zero coupon bond, a twenty-year zero coupon bond, an equity index and a real estate index. The bonds are annually replaced such that the time to maturity remains constant.

The model is described by the following equations, where all financial variables are nominal amounts, unless otherwise stated. The discounted cash flow at $t=1,\dots,T$ is given by
\[ \zeta_t = \begin{cases} D_t\max(A_t, G_t)/C_t, & t < T \\
L_{T-1}\max(A_T, G_T)/C_T, & t = T \end{cases}
\]
where
\begin{multicols}{2}
\begin{itemize}
    \item[--] $D_t$: total dead in period $(t-1, t]$
    \item[--] $L_t$: total of policyholders alive at time t
    \item[--] $A_t$: value of assets at time $t$ (per policy)
     \item[--] $G_t$: guaranteed value at time $t$  (per policy)
    \item[--] $C_t$: value of the cash account at time $t$
    \item[--] $T$: maturity date of the policies
\end{itemize}
\end{multicols}

The value of assets at $t=0,\dots,T$ is given by
\[ A_t = \sum_{i=1}^4 U_t^{i}V_t^{i} =  \sum_{i=1}^4 U_{t-1}^{i}V_{t}^{i-} + P_t,  \]
where
\[ V_t^{i} = \begin{cases} B(t, t+10), & i=1, \\
B(t, t+20) ,& i=2, \\
{S}_t, & i=3, \\
RE_t, & i=4 ,
\end{cases}\qquad \text{and}\qquad V_t^{i-} = \begin{cases} B(t, t+9), & i=1 ,\\
B(t, t+19) ,& i=2, \\
V_t^{i}, & i=3,4,
\end{cases}\]
denotes the unit price of asset $i$ at time $t$, where we use the notation $V^{i-}_t$ to express the rolling over of the constant-maturity bond investments for $i=1,2$, and
\begin{multicols}{2}
\begin{itemize}
    \item[--] $U^i_t = (U_{t-1}^{i}V_{t}^{i-} + M_i P_t) /V^i_t$: number of units of asset $i$ held in period $(t, t + 1]$ (per policy), where $U^{i}_{-1}:=0$
    \item[--] $B(t, s)$: value at time $t$ of a bond maturing at time $s$
    \item[--] ${S}_t$: value of equity index at time $t$
    \item[--] $RE_t$: value of real estate at time $t$
    \item[--] $M_i$: asset allocation mix, henceforth fixed to $M=\Big(\frac{1}{3}, \frac{1}{3}, \frac{1}{5}, \frac{2}{15}\Big)$
    \item[--] $P_t$: premium paid at $t$ for period $(t,t+1]$ (per policy)
\end{itemize}
\end{multicols}

The policy variables are given by
\[ P_t=100,\qquad\text{and}\qquad G_t =\begin{cases} 0,& t=0 \\  G_{t-1}+P_{t-1},& t\ge 1\end{cases}\]

As for the demographic variables, the total dead and total alive are given by
\[ D_t=\sum_x D^x_t,\qquad\text{and}\qquad  L_t =\sum_x L^x_t, \]
where
\begin{itemize}
    \item[--] $D_t^x=L_{t-1}^x q_x(t)$: total dead of age $x$ at $t-1$ in period $(t-1, t]$
    \item[--] $L_t^x=L_{t-1}^x-D^x_t$: total alive of age $x$  at time $t$, with $L_0^x = 1000$ for all $x\in(30,70)$
    \item[--] $q_x(t)$: death rate for age $x$ at $t-1$ in period $(t-1, t]$.
\end{itemize}

In total the stochastic driver $X$ has $d=5$ components: two for the interest rate model, one for the equity model, one for the real estate model, and one for the stochastic mortality. The interest rate and equity models, for $B$ and ${S}$, are those described in Appendix~\ref{ESG app} and used in previous examples. The real estate model, for $RE$, is the same as the equity model from Appendix~\ref{ESG app}, but uses an independent stochastic driver and a lower volatility than the equity model. The stochastic mortality follows the Lee--Carter model (\cite{leecarter}) to provide a trend and random fluctuations over time. More specifically, we model the death rate as
\begin{align*}
q_x(t)&=1-e^{-m_x(t)}           &  m_x(t) &= e^{a_x+b_xk(t)}\\
k(t) &= k(t-1) - 0.365 + \epsilon_t         &  \epsilon_t &= 0.621X_t^{(lc)} \\
k(0)&=-11.41
\end{align*}
where
\begin{itemize}
    \item[--] $m_x(t)$: force of mortality at time $t$ for age $x$
    \item[--] $X_t^{(lc)}$: component of the stochastic driver $X$ at time $t$ used for mortality model
    \item[--] $a_x$ and $b_x$: Lee--Carter parameters (table in Section \ref{subsec:lee-carter params}).
\end{itemize}

\subsection{Results}\label{sec:varann_results}
The results for the variable annuity guarantee confirm those of the European call option example: the replicating martingale works very well, in particular the neural network model, which provides the best results in most cases. However, the more complex example also shows some limitations of the methods.

In the estimation of the present value, Table~\ref{tab:varann_mae_pv_all_methods} shows that nested Monte Carlo (nMC) is still very effective, but regression-based methods provide slightly better accuracy. The neural network model performs relatively badly in the case with the lowest number of samples (1,000) and high dimensions ($T=40$), providing the worst results in that case. This is caused by over-fitting, as we describe in the analysis of the mean relative $L_1$ error below. Indeed, as for the European option example, by cross-validation we have chosen $m=101$ nodes, whereof one is a bias term only, say $A_{101}=0$. The total number of parameters is thus $2{,}600+101=2{,}701$ for $T=5$, and $20{,}100+101=20{,}201$ for $T=40$. Alternative specifications of the width of the neural network are discussed in Appendix~\ref{subsec:sensitivity_nn}. The quality reaches that of the other methods as the number of samples increase. Finally, we observe that the polynomial LDR method---which is calculated with $p=10$---shows its advantage over the full polynomial basis not only in being able to solve the high dimensional case, but also in the estimation of the low dimensional case with low number of samples. The full polynomial basis has a MApE of 61\% due to the basis containing 3,276 elements, see Table \ref{tabdimPol}, which exceeds the 1,000 available samples. The polynomial LDR has a MApE of less than 0.1\% due to only containing 286 basis elements.

\begin{table}[ht]
    \centering
    \small
    \caption{Insurance liability, comparison of present value MApE (in percentage points)}
\begin{tabular}{@{}rrrrrrrr@{}}
\toprule
\multicolumn{2}{r}{\textbf{}} & \multicolumn{2}{r}{\textbf{Full Polynomial basis}} &  \textbf{LDR} & \multicolumn{2}{r}{\textbf{Neural Network}} \\
\textbf{Samples} & \textbf{nMC} &                    Regress-now & Regress-later & Regress-later &             Regress-now & Regress-later \\
\midrule
\multicolumn{1}{l}{\textbf{\textbf{\textbf{T}}=5}}  &   &   &   &   &   &   \\
           1,000 &          0.3 &                            0.2 &          61.0 &           \textless{}0.1 &                     0.3 &           0.1 \\
           5,000 &          0.3 &                            0.1 &           \textless{}0.1 &           \textless{}0.1 &                     0.1 &           \textless{}0.1 \\
          10,000 &          0.2 &                            0.1 &           \textless{}0.1 &           \textless{}0.1 &                     0.1 &           \textless{}0.1 \\
          50,000 &          0.1 &                            \textless{}0.1 &           \textless{}0.1 &           \textless{}0.1 &                     \textless{}0.1 &           \textless{}0.1 \\
\multicolumn{1}{l}{\textbf{\textbf{\textbf{T}}=40}}  &   &   &   &   &   &   \\
           1,000 &          0.5 &                            0.5 &             &           0.2 &                     0.6 &           6.1 \\
           5,000 &          0.3 &                            0.2 &             &           0.1 &                     0.2 &           0.3 \\
          10,000 &          0.3 &                            0.2 &             &           0.1 &                     0.2 &           0.4 \\
          50,000 &          0.2 &                            0.1 &             &           \textless{}0.1 &                     0.1 &           0.1 \\
\bottomrule
\end{tabular}
    \label{tab:varann_mae_pv_all_methods}
\end{table}

In the estimation of the expected shortfall, shown in Table~\ref{tab:varann_mae_es_all_methods} and Figure~\ref{fig:box_plot_var_ann_es}, and the analysis of the mean relative $L_1$ error, shown in Table~\ref{tab:varann_l1_error_all_methods}, we can observe that regress-later methods dominate over regress-now methods and nested Monte Carlo, with better mean absolute error and standard deviation. Unlike the case in the European call example where neural networks completely dominated the quality comparison, polynomial LDR shows better results in a few cases. However, which method shows better results is very sensitive to the choice of hyper-parameters. We provide a sensitivity analysis for hyper-parameters in Appendix~\ref{sec:sensitivity}. Overall, neural networks have more room for improvement with an alternative choice of hyper-parameters and can be assumed to produce better results in this variable annuity example. We can observe several cases where an insufficient number of training samples leads to over-fitting and poor out-of-sample results. For example, for the full polynomial basis and $T=5$, we find a large improvement in results when the training data changes from 1,000 samples to 5,000 samples. This basis has 3,276 elements, see Table \ref{tabdimPol}, which means that when working with 1,000 samples we have more parameters than samples. The same effect can be seen in the neural network replicating martingale for $T=40$ when the sample size changes from 10,000 to 50,000 samples. This can be explained by the fact that this model has 20,201 parameters, as mentioned above. In some cases, for example, the case neural network regress-later estimator for $T=5$ the MApE ES increases when the sample size increases from 5,000 to 10,000 and 50,000, see Table~\ref{tab:varann_mae_es_all_methods}. This behaviour is not present in the relative mean $L_1$ error, as evidenced in Table~\ref{tab:varann_l1_error_all_methods}. This is due to the divergence between the error being minimized---errors along the full cash flows distribution---and the error being measured---errors in the tail of $T=1$ conditional expectation distribution.

\begin{table}[ht]
    \centering
    \small
    \caption{Insurance liability, comparison of expected shortfall MApE (in percentage points)}
\begin{tabular}{@{}rrrrrrrr@{}}
\toprule
\multicolumn{2}{r}{\textbf{}} & \multicolumn{2}{r}{\textbf{Full Polynomial basis}} &  \textbf{LDR} & \multicolumn{2}{r}{\textbf{Neural Network}} \\
\textbf{Samples} & \textbf{nMC} &                    Regress-now & Regress-later & Regress-later &             Regress-now & Regress-later \\
\midrule
\multicolumn{1}{l}{\textbf{\textbf{\textbf{T}}=5}}  &   &   &   &   &   &   \\
           1,000 &         30.7 &                           80.6 &         613.9 &           7.9 &                   198.1 &           2.9 \\
           5,000 &         11.1 &                           18.6 &           0.5 &           5.3 &                    48.7 &           0.5 \\
          10,000 &          9.2 &                           10.3 &           0.3 &           5.1 &                    25.8 &           0.6 \\
          50,000 &          5.7 &                            3.3 &           0.2 &           3.7 &                     6.6 &           0.8 \\
\multicolumn{1}{l}{\textbf{\textbf{\textbf{T}}=40}}  &   &   &   &   &   &   \\
           1,000 &        105.4 &                          226.4 &             &          22.9 &                   520.7 &          14.4 \\
           5,000 &         32.7 &                           64.9 &             &           5.7 &                   147.3 &          11.1 \\
          10,000 &         18.7 &                           35.9 &             &           5.9 &                    84.4 &          10.5 \\
          50,000 &         10.5 &                           10.0 &             &           7.2 &                    13.2 &           0.5 \\
\bottomrule
\end{tabular}

    \label{tab:varann_mae_es_all_methods}
\end{table}

\begin{table}[ht]
    \centering
    \small
    \caption{Insurance liability, comparison of relative mean $L_1$ error (in percentage points) }
\begin{tabular}{@{}rrrrrrrr@{}}
\toprule
       \textbf{} & \multicolumn{2}{r}{\textbf{Full Polynomial basis}} &  \textbf{LDR} & \multicolumn{2}{r}{\textbf{Neural Network}} \\
\textbf{Samples} &                    Regress-now & Regress-later & Regress-later &             Regress-now & Regress-later \\
\midrule
\multicolumn{1}{l}{\textbf{\textbf{\textbf{T}}=5}}  &   &   &   &   &   \\
           1,000 &                            1.6 &          61.0 &           0.2 &                     4.7 &           0.1 \\
           5,000 &                            0.7 &           \textless{}0.1 &           0.2 &                     1.6 &           0.1 \\
          10,000 &                            0.5 &           \textless{}0.1 &           0.2 &                     1.0 &           \textless{}0.1 \\
          50,000 &                            0.2 &           \textless{}0.1 &           0.2 &                     0.3 &           \textless{}0.1 \\
\multicolumn{1}{l}{\textbf{\textbf{\textbf{T}}=40}}  &   &   &   &   &   \\
           1,000 &                            3.3 &             &           0.8 &                    10.0 &           6.1 \\
           5,000 &                            1.4 &             &           0.3 &                     3.4 &           0.6 \\
          10,000 &                            1.0 &             &           0.4 &                     2.2 &           0.5 \\
          50,000 &                            0.4 &             &           0.4 &                     0.6 &           0.1 \\
\bottomrule
\end{tabular}

    \label{tab:varann_l1_error_all_methods}
\end{table}

\begin{figure}[ht]
\includegraphics[width=\textwidth]{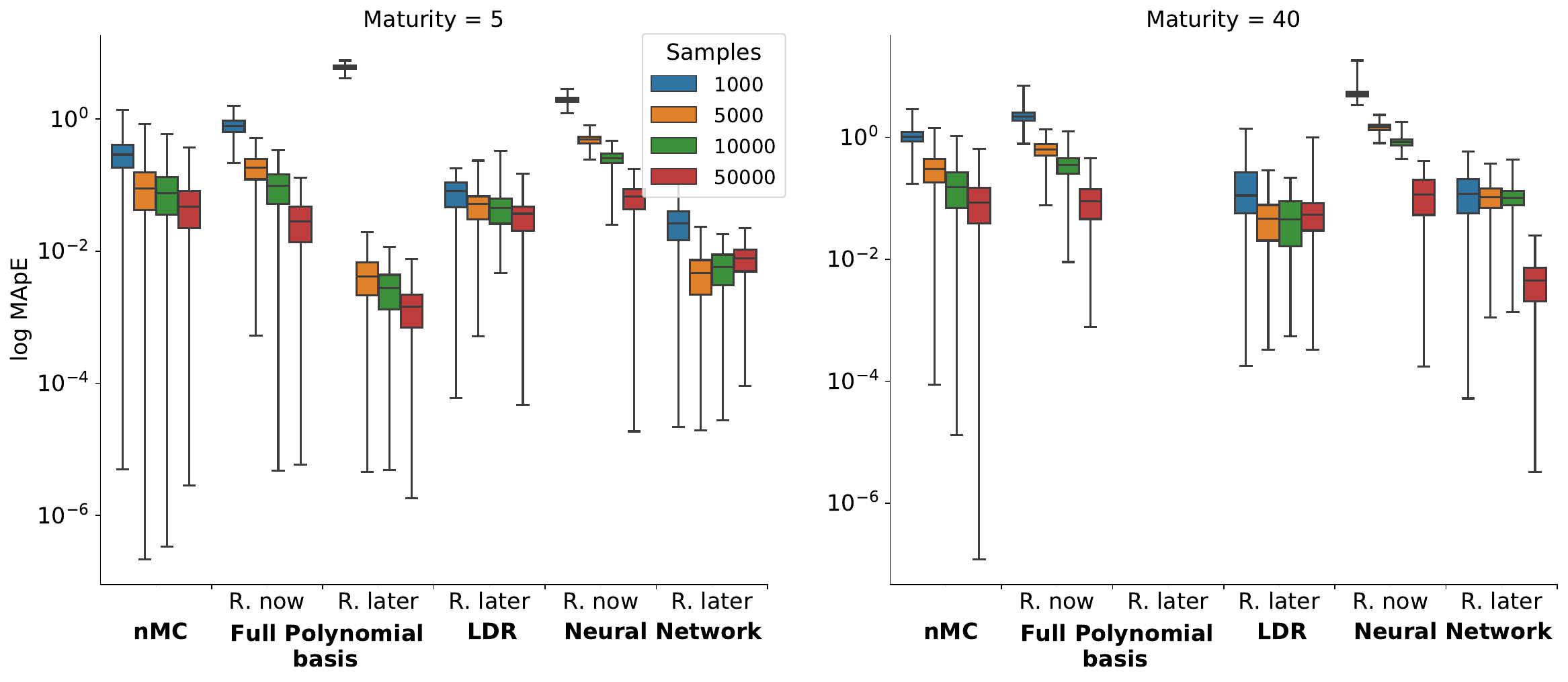}
\caption{Distribution of expected shortfall estimates per method for the insurance liability example. Boxes show the upper and lower quartiles of the empirical distributions, while whiskers show their maxima and minima. Due to the logarithmic scale, a visual artifact is introduced by which the inter-quartile range seems to increase with larger sample sizes---most notably in the nested Monte Carlo (nMC) for Maturity=40. In reality, the inter-quartile range decreases with larger sample sizes, but the logarithmic scale makes the box bigger for smaller errors.}
\label{fig:box_plot_var_ann_es}
\end{figure}

A comparison of the runtimes for the insurance example is important to determine the relative strength of the methods as feasible solution in the real world. Details are given in Appendix~\ref{app_runtime}.

\section{Conclusion}\label{sec_conc}

In the context of the need for accurate and fast calculations in portfolio pricing and risk management, we have introduced a data-driven method to build replicating martingales under different functional basis. This method yields lower errors than standard nested Monte Carlo simulation.

The model learns the features necessary for an effective low-dimensional representation from finite samples in a supervised setting. By doing so, it can be very effective in high-dimensional problems, without some of the usual difficulties associated with them.

We have presented two examples to demonstrate the usefulness of replicating martingales in the calculation of economic capital. The first is a typical benchmark example for calculations involving financial derivatives: a European call option. The second is a path-dependent insurance product, a variable annuity guarantee. Replicating martingales outperform other methods in the literature and in use in the financial industry for these two representative cases. This is illustrated by extensive comparisons and sensitivity analyses.

\clearpage
\section*{Appendix}
\appendix

\section{Quality metrics} \label{sec:qual_metrics}
Since the focus on this paper are applications in pricing and risk managements, we use two key quality metrics. The first one looks into the goodness of fit in the tail of the distribution, and the second one the goodness of fit across the body of the distribution. For the tail of the distribution we look at expected shortfall and value at risk, for the loss-making tail. For the body of the distribution we look at the $L_1$ error.

We treat the models as statistical estimators since their estimates are subject to the randomness of their inputs. For that reason, for each those metrics described above we derive an empirical distribution based on $R$ macro-runs of the the entire simulation-estimation-prediction chain of calculations. That means that we also need to define which metric summarizes the results of the empirical distribution. In both cases (tail error and $L_1$ error) we use the mean absolute error. In all cases we work with relative errors, expressed as a percentage. Root mean squared errors would have been an option but the advantages of the mean absolute error have been well documented in \cite{willmott2005advantages} and \cite{chai2014root}.

In the sections below we describe in detail the calculation of our two quality metrics: mean absolute percentage error on tail error (MApE), and mean relative $L_1$ error.

\subsection{Mean absolute percentage error}
Let us consider an empirical distribution of $X_{1:t}$, composed of $n$ samples. From this distribution we can obtain an empirical distribution of $V_t$. Given a function $f^*$, we obtain $R$ repetitions of its finite sample estimator $\widehat{f}$. For each function in these $\{\widehat{f}_j\}_{j=1}^R$, we can produce an empirical distribution of its value estimator $\widehat{V}_t=\E_t^\Q[\widehat{f}(X)]$ using $X_{1:t}$, therefore obtaining a set of empirical distributions $\{\widehat{V}_t^{(j)}\}_{j=1}^R$.

Given a benchmark expected shortfall calculation at $\alpha$ (e.g., $\alpha=99\%$) confidence, $\mbox{ES}_\alpha[-\Delta V_t]$, an estimator of such quantity ${\mbox{ES}_\alpha[-\Delta \widehat{V}_t^{(j)}]}$, and $R$ repetitions (independent samples) of such estimator $j=1\dots R$, the mean absolute percentage error (MApE) is defined as

\[ \mbox{MApE ES} = \frac{1}{R} \sum_{j=1}^R \frac{| {{\mbox{ES}_\alpha[-\Delta \widehat{V}_t^{(j)}]}} - \mbox{ES}_\alpha[-\Delta V_t] |}{ \mbox{ES}_\alpha[-\Delta V_t]}.
\]

The MApE metric can also be applied to the present value of $V_t$, $\E[V_t]=V_0$:

\[ \mbox{MApE PV} = \frac{1}{R} \sum_{j=1}^R \frac{| \widehat{V}_0^{(j)} - V_0 |}{ V_0}.
\]

\subsection{Mean relative $L_1$ error}

Given the above, the mean relative $L_1$ error is defined as

\[ \frac{1}{R} \sum_{j=1}^R \frac{ \E[|\widehat{V}_t^{(j)} - V_t|]}{ \E[|  V_t|]}\]
This metric is related to the error on the expected shortfall in the following way

\[ | \mbox{ES}_\alpha[\hat{V}_t] - \mbox{ES}_\alpha[V_t] | \leq \frac{1}{\alpha}\E[|\hat{V}_t-V_t|].\]
This shows that for any $\alpha$, the expected shortfall MApE is bounded by a multiple of the $L_1$ error. While the MApE ES is a metric calculated for a particular $\alpha$ and only takes into account the distribution beyond the $\alpha$-th percentile, the mean relative $L_1$ error takes into account the whole distribution and bounds the expected shortfall error for any $\alpha$.

\section{Economic scenario generator}\label{ESG app}

We describe the basic financial models underlying the economic scenario generator of the examples in this paper. We assume that the stochastic driver is multinormal $\vect(X)\sim N(0,I_{dT})$, and consider either $d=3$ or $d=5$.

\subsection{Interest rate model}
The interest rate model is based on the continuous time Hull--White short rate model \[dr_t = \kappa (b(t)-r_t)dt+\sigma dW_t,\] for parameters $\kappa$, $\sigma$, and function $b(t)$, where $W$ denotes a Brownian motion under the risk-neutral measure, see, e.g., \cite{glasserman2013monte}. The nominal price at time $t$ of a zero-coupon bond with maturity $T$ is given by
\[    B(t,T)=\exp(-A(t,T)r_t+C(t,T)) \]
where
\begin{align*}
    A(t,T)&=\frac{1}{\kappa} (1-e^{-\kappa(T-t)}),\\
    C(t,T)&=-\kappa h(t,T)
    +\frac{\sigma^2}{2\kappa^2}\Big[(T-t)+\frac{1}{2\kappa}(1-e^{-2\kappa(T-t)})+\frac{2}{\kappa}(e^{-\kappa(T-t)}-1)\Big]
\end{align*}
and where we denote $h(t,T)=\int_t^T \int_t^u e^{-\kappa(u-s)}b(s)dsdu$.

In discrete time, we exactly simulate the short rate, $r_t$, and log-cash account, $Y_t=\int_0^t r_u du$, jointly from the above Hull--White model according to the formulas in \cite{glasserman2013monte}, which are based on a two-dimensional Gaussian stochastic driver. We therefore use the first two components of $X$, that is, $X_{1,t}$ and $X_{2,t}$, as follows. As for $r_t$, define
\begin{align*}
    g(t)&=\int_t^{t+1} e^{-\kappa(t+1-s)}b(s)ds \\
    \sigma_r &=\frac{\sigma^2}{2\kappa}(1-e^{-2\kappa})
\end{align*}   
and set $r_{t+1}= e^{-\kappa} r_t + \kappa g(t)+ \sigma_r  X_{1,t+1}$.

As for $Y_t$, define
\begin{align*}
  h(t)&=h(t,t+1)\\
  \sigma_Y &= \frac{\sigma^2}{\kappa^2}\Big(1+\frac{1}{2\kappa}(1-e^{-2\kappa}) +\frac{2}{\kappa}(e^{-\kappa}-1)\Big) \\ 
    \sigma_{rY}  &=\frac{\sigma^2}{2\kappa}(1+e^{-2\kappa}-2e^{-\kappa})  \\
    \rho_{rY} &=  \sigma_{rY}/(\sigma_r\sigma_Y) 
\end{align*}
and set $Y_{t+1} = Y_t + (1/\kappa)(1-e^{-\kappa})r_t+\kappa h(t) + \sigma_Y X'_{2,t+1}$, for the correlated driver $X'_{2,t+1}=\rho_{rY}X_{1,t+1}+\sqrt{1-\rho^2_{rY}}X_{2,t+1}$.

\subsection{Equity and real estate index models}

For a given matrix $\Sigma$ that encodes the desired correlations, we denote the correlated Gaussian stochastic driver $X'=\Sigma X$. For the examples in this paper, we set $X'_{1,t}=X_{1,t}$, $X'_{2,t}$ as above, $X'_{3,t}$, $X'_{4,t}$ to be correlated with $X_{1,t}$, and $X'_{5,t}=X_t^{(lc)}=X_{5,t}$---used for the mortality model---to be independent of all other variables.

For both, equity and real estate, a geometric Brownian process models the respective index excess return, with the recursive formula
\[Z_{j,t}=Z_{j,t-1}\exp \left( - \sigma_j ^{2}/2+\sigma_j X'_{ j,t}\right),\]
where $j=3$ for the equity index and $j=4$ for the real estate index. The equity index ${S}_t$ and the real estate index $RE_t$ are then given by ${S}_t=C_t Z_{3,t}$ and $RE_t=C_t Z_{4,t}$, respectively, where $C_t=\exp(Y_t)$ denotes the cash account.

\subsection{Lee--Carter parameters} \label{subsec:lee-carter params}

The Lee--Carter parameters for the mortality model are based on the findings in the original paper \cite{leecarter}, and they are shown in Table~\ref{tab:lee--carter}.

\begin{table}[H]
\caption{Lee--Carter parameters $a_x$ and $b_x$ for every age $x$}
\centering
\begin{tabular}{@{}lll@{}}
\toprule
\textbf{x}                & $\mathbf{a_x}$ & $\mathbf{b_x}$ \\ \midrule
0                         & -3.641090     & 0.90640       \\
(1, 2, 3, 4)              & -6.705810     & 0.11049       \\
(5, 6, 7, 8, 9)           & -7.510640     & 0.09179       \\
(10, 11, 12, 13, 14)      & -7.557170     & 0.08358       \\
(15, 16, 17, 18, 19)      & -6.760120     & 0.04744       \\
(20, 21, 22, 23, 24)      & -6.443340     & 0.05351       \\
(25, 26, 27, 28, 29)      & -6.400620     & 0.05966       \\
(30, 31, 32, 33, 34)      & -6.229090     & 0.06173       \\
(35, 36, 37, 38, 39)      & -5.913250     & 0.05899       \\
(40, 41, 42, 43, 44)      & -5.513230     & 0.05279       \\
(45, 46, 47, 48, 49)      & -5.090240     & 0.04458       \\
(50, 51, 52, 53, 54)      & -4.656800     & 0.03830       \\
(55, 56, 57, 58, 59)      & -4.254970     & 0.03382       \\
(60, 61, 62, 63, 64)      & -3.856080     & 0.02949       \\
(65, 66, 67, 68, 69)      & -3.473130     & 0.02880       \\
(70, 71, 72, 73, 74)      & -3.061170     & 0.02908       \\
(75, 76, 77, 78, 79)      & -2.630230     & 0.03240       \\
(80, 81, 82, 83, 84)      & -2.204980     & 0.03091       \\
(85, 86, 87, 88, 89)      & -1.799600     & 0.03091       \\
(90, 91, 92, 93, 94)      & -1.409363     & 0.03091       \\
(95, 96, 97, 98, 99)      & -1.036550     & 0.03091       \\
(100, 101, 102, 103, 104) & -0.680350     & 0.03091       \\
(105, 106, 107, 108)      & -0.341050     & 0.03091       \\ \bottomrule
\end{tabular}
\label{tab:lee--carter}
\end{table}

\section{Proofs}
\label{sec:proof_thmVk}

This section contains all proofs and some auxiliary results of independent interest.

\subsection{Proof of Lemma \ref{lemexi}}
In view of \eqref{eqnbetaAnew} and by orthogonality, we have $\| f -\phi_\theta^\top\beta\|^2_{L^2_\Q} \ge \| f -\phi_\theta^\top\beta_\theta\|^2_{L^2_\Q} = \|f\|_{L^2_\Q}^2 - \langle f,\phi_\theta^\top\beta_\theta\rangle_{L^2_\Q}$, for all $(\theta,\beta)\in\Theta\times\R^m$. On the other hand, by assumption \ref{lemexi3} we can write $\beta_\theta= \langle\phi_\theta,\phi_\theta^\top\rangle_{L^2_\Q}^{-1} \langle \phi_\theta,f\rangle_{L^2_\Q}$ and hence $\langle f,\phi_\theta^\top\beta_\theta\rangle_{L^2_\Q}  = \langle f,\phi_\theta^\top\rangle_{L^2_\Q}\langle\phi_\theta,\phi_\theta^\top\rangle_{L^2_\Q}^{-1} \langle \phi_\theta,f\rangle_{L^2_\Q}$. Hence \eqref{optthebet} is equivalent to \eqref{eqnbetaAnew} and
\[ \max_{\theta \in \Theta } \langle f,\phi_\theta^\top\rangle_{L^2_\Q}\langle\phi_\theta,\phi_\theta^\top\rangle_{L^2_\Q}^{-1} \langle \phi_\theta,f\rangle_{L^2_\Q}.\]
By the assumptions of the lemma, $\theta\mapsto \langle f,\phi_\theta^\top\rangle_{L^2_\Q}\langle\phi_\theta,\phi_\theta^\top\rangle_{L^2_\Q}^{-1} \langle \phi_\theta,f\rangle_{L^2_\Q}$ is continuous, and hence attains its maximum on the compact set $\Theta$. This completes the proof.

\subsection{Proof of Theorem \ref{thmVk}}

Theorem~\ref{thmVk} follows from Lemmas~\ref{lemImAphiA} and \ref{lemlinalgA} below.

\begin{lemma}\label{lemImAphiA}
Let $A,\tilde A\in\R^{dT\times p}$ with full rank and $b,\tilde b\in\R^p$. The following are equivalent:
\begin{enumerate}
  \item\label{lemImAphiA1} $\ker A^\top =\ker \tilde A^\top $
  \item\label{lemImAphiA1a} $ \tilde A = A S $, for some invertible $p\times p$-matrix $S$
  \item\label{lemImAphiA2} $\spn\{\phi_{(A,b),1},\dots,\phi_{(A,b),m}\}=\spn\{\phi_{(\tilde A,\tilde b),1},\dots,\phi_{(\tilde A,\tilde b),m}\}$
\end{enumerate}
\end{lemma}

\begin{proof} 
\ref{lemImAphiA1}$\Leftrightarrow$\ref{lemImAphiA1a}: is elementary.

\ref{lemImAphiA1a}$\Rightarrow$\ref{lemImAphiA2}: define $p_i(y)=g_i(y+b)$ and $\tilde p_i(y)=g_i(S^\top y+\tilde b)$ and note that $\{ p_1,\dots, p_m\}$ as well as $\{\tilde p_1,\dots,\tilde p_m\}$ forms a basis of ${\rm Pol}_\delta(\R^{p})$. On the other hand, we have $\phi_{(A,b),i}(x)=g_i(A^\top x+ b) = p_i(A^\top x)$ and $\phi_{(\tilde A,\tilde b),i}(x)=g_i(S^\top A^\top x+\tilde b)=\tilde p_i(A^\top x)$. This yields the claim.

\ref{lemImAphiA2}$\Rightarrow$\ref{lemImAphiA1}: we argue by contradiction and assume that there exists some $\tilde x\in\ker\tilde A^\top\setminus\ker A^\top$. Then $y\mapsto y^\top A^\top\tilde x$ is in ${\rm Pol}_\delta(\R^{p})$ and can thus be written as linear combination $y^\top A^\top\tilde x=\sum_{i=1}^m c_i g_i(y+b)$. We
obtain that the function $h(x)= x^\top AA^\top \tilde x=\sum_{i=1}^m c_i\phi_{(A,b),i}(x)$ lies in the $\spn\phi_{(A,b)}$. But $h\notin\spn\phi_{(\tilde A,\tilde b)}$ because $\phi_{(\tilde A,\tilde b),i}(t\tilde x)=g_i(\tilde b)$ is constant in $t\in\R$, for all $i$, while $h(t\tilde x)=\|A^\top\tilde x\|^2 t$ is not. This completes the proof.
\end{proof}

\begin{lemma}\label{lemlinalgA}
Let $A,\tilde A\in\R^{dT\times p}$ with full rank. The following are equivalent:
\begin{enumerate}
  \item\label{lemlinalgA1} $\tilde A\in V_p(\R^{dT})$ and $ \tilde A = A S $, for some invertible $p\times p$-matrix $S$
  \item\label{lemlinalgA2} $\tilde A=A(A^\top A)^{-1/2}U$ for some orthogonal $p\times p$-matrix $U$
  \end{enumerate}
\end{lemma}

\begin{proof}
\ref{lemlinalgA1}$\Rightarrow$\ref{lemlinalgA2}: we obtain $I_p=\tilde A^\top\tilde A=S^\top A^\top A S=S^\top(A^\top A)^{1/2} (A^\top A)^{1/2}S$. Hence the $p\times p$-matrix $U=(A^\top A)^{1/2}S$ is orthogonal and $S=(A^\top A)^{-1/2}U$, which yields the claim.

\ref{lemlinalgA2}$\Rightarrow$\ref{lemlinalgA1}: is elementary.
 \end{proof}

\subsection{Proof of Theorem \ref{thmpolyuniqueN}}

As $V_p(\R^{dT})$ is a compact manifold, it follows by inspection that the assumptions of Lemma \ref{lemexi} are met. Hence there exists a minimizer of \eqref{optthebet}. The non-uniqueness statement is proved by means of the following counterexample. Assume $f(x)=f(Vx)$ and the pushforward $V_\ast\Q=\Q$ for some orthogonal $dT\times dT$-matrix $V$. Then, for any $A\in V_p(\R^{dT})$ and $\beta\in\R^m$, we have $\|f - \phi_{A}^\top\beta\|_{L^2_\Q} = \| f - \phi_{V^\top A}^\top\beta\|_{L^2_\Q}$. But $\spn\{\phi_{A,1},\dots,\phi_{A,m}\}\neq \spn\{\phi_{V^\top A,1},\dots,\phi_{V^\top A,m}\}$ in general by Theorem \ref{thmVk}. This completes the proof of Theorem \ref{thmpolyuniqueN}.

\subsection{Proof of Theorem \ref{thmReLUspanNew}}

We follow the heuristic arguments of \cite{he_etal_20}.\footnote{We also complete some arguments in \cite{he_etal_20}, who do not explain what kind of derivative ``$\nabla \sum_{i=1}^m c_j\phi_{(a_j,b_j)}$'' stands for.}  First, note that any linear combination $h=\sum_{j=1}^m c_j \phi_{( a_j, b_j)}$ is a continuous, piece-wise affine function. As such it is \say{Bouligand differentiable} on $\R^{dT}$, see \cite[Theorem 3.1.2]{sch_12}. That is, its directional derivative $\nabla_v h(x)=\lim_{\epsilon\downarrow 0} (h(x+\epsilon v)-h(x))/\epsilon$ exists for all $x,v\in\R^{dT}$, and it provides a first order approximation, $\lim_{y\to x} \|h(y)- h(x)-\nabla_{y-x} h(x)\|/\|y-x\|=0$. Accordingly, the classical calculus rules carry over and we have $\nabla_v \sum_{j=1}^m c_j \phi_{( a_j, b_j)}(x) = \sum_{j=1}^m c_j \nabla_v\phi_{( a_j, b_j)}(x)$, see \cite[Corollary 3.1.1]{sch_12}.

Next, for a function $h:\R^{dT}\to\R$, we denote by $D_h$ the set of points of discontinuity. For $h(x)=\nabla_v\phi_{(a_i,b_i)}(x)$ and any $v\in\R^{dT}$ we obtain
\[  D_{ \nabla_v \phi_{(a_i,b_i)}}=\begin{cases}
\emptyset,& \text{if $v^\top a_i=0$,}\\
  H_{(a_i,b_i)},& \text{otherwise,}
\end{cases}\]
for the affine hyperplane $H_{(a_i,b_i)}=\{ x\mid a_i^\top x+b_i=0\}$.

Now assume $\phi_{(a_i,b_i)}\in \spn\{ \phi_{(\tilde a_1,\tilde b_1)},\dots,\phi_{(\tilde a_m,\tilde b_m)}\}$, so that $\phi_{(a_i,b_i)}=\sum_{j=1}^m c_j \phi_{(\tilde a_j,\tilde b_j)}$ for some real coefficients $c_j$. Then
\[H_{(a_i,b_i)}= D_{ \nabla_v \phi_{(a_i,b_i)}}=D_{ \sum_{j=1}^m c_j \nabla_v\phi_{(\tilde a_j,\tilde b_j)} }\subseteq\cup_{j=1}^m D_{ \nabla_v \phi_{(\tilde a_j,\tilde b_j)}}\subseteq \cup_{j=1}^m H_{(\tilde a_j,\tilde b_j)},\]
where we used the obvious relation $D_{t_1 h_1+t_2 h_2}\subseteq D_{h_1}\cup D_{h_2}$, for functions $h_1,h_2$ and real coefficients $t_1,t_2$. This implies that $(a_i,b_i)\in\{\pm (\tilde a_1,\tilde b_1),\dots, \pm (\tilde a_m,\tilde b_m)\}$. Since $i$ was arbitrary, we obtain $\{\pm (a_1,b_1),\dots, \pm (a_m,b_m)\}\subseteq\{\pm (\tilde a_1,\tilde b_1),\dots, \pm (\tilde a_m,\tilde b_m)\}$. A similar argument for $(\tilde a_i,\tilde b_i)$ in lieu of $(a_i,b_i)$ implies the converse inclusion. This proves \ref{thmReLUspanNew1}.

For the proof of \ref{thmReLUspanNew2} we argue by contradiction. Suppose \eqref{thmReLUspanNew2eq} does not hold, so that $\phi_{(a_i,b_i)}\in\spn\{ \phi_{(a_j,b_j)}\mid j\neq i\}$ for some $i$. Hence \eqref{thmReLUspanNew1eq} holds for $(\tilde a_j,\tilde b_j):=( a_j, b_j)$ for all $j\neq i$, and $(\tilde a_i,\tilde b_i):=(a_j,b_j)$ for some $j\neq i$. But then part \ref{thmReLUspanNew1} implies that $(a_i,b_i)=\pm (a_j,b_j)$ for some $j\neq i$, which contradicts the assumption of \ref{thmReLUspanNew2}. This completes the proof of Theorem \ref{thmReLUspanNew}.

\begin{remark}
One may reckon that \eqref{thmReLUspanNew1eq} and \eqref{thmReLUspanNew2eq} together imply 
\[\{  (a_1,b_1),\dots,   (a_m,b_m)\}=\{  (\tilde a_1,\tilde b_1),\dots,   (\tilde a_m,\tilde b_m)\}.\]
However, this is not true in general. Indeed, let $(a_1,b_1),\dots,(a_3,b_3)\in \Scal_{dT}$ be as in Example~\ref{exlinind}, and define $S=\{(a_1,b_1),\pm(a_2,b_2),\pm(a_3,b_3)\}$ and $\tilde S=\{-(a_1,b_1),\pm(a_2,b_2),\pm(a_3,b_3)\}$. Then $F=\{\phi_{(a,b)}\mid (a,b)\in S\}$ and $\tilde F=\{\phi_{(a,b)}\mid (a,b)\in \tilde S\}$ are both linearly independent sets. On the other hand, we have $c_1\phi_{-(a_1,b_1)} = c_1 \phi_{(a_1,b_1)} + \sum_{i=2}^3 c_i ( \phi_{(a_i,b_i)}-\phi_{-(a_i,b_i)})$, and hence $\spn F =\spn \tilde F$. But $S\neq \tilde S$.
\end{remark}

\subsection{Proof of Theorem \ref{thmNEReLU}}

We prove the theorem by means of two counterexamples. First, let $d=T=1$, $\Q=N(0,1)$, and $f(x)=1_{[0,\infty)}(x)$. For the feature map, we let $m=2$, and set $b_n=1/n$, $a_n=\sqrt{1-b_n^2}$, and $\theta_n=((a_n,b_n),(0,1))\in (\Scal_1)^2$. Then $\theta_n\to ((0,1),(0,1))$ as $n\to\infty$. On the other hand, for $\beta_n=(1/b_n,-a_n/b_n)^\top$, we have
  \[   \phi_{\theta_n}(x)^\top\beta_n = \frac{1}{b_n}(a_n x+b_n)^+  -\frac{a_n}{b_n}x^+ \to 1_{[0,\infty)}(x) \quad\text{in $L^2_\Q$ as $n\to\infty$.}\]
Hence $\inf_{(\theta,\beta)\in (\Scal_1)^2\times\R^2} \|f - \phi_\theta^\top\beta\|_{L^2_\Q} = \lim_{n} \|f - \phi_{\theta_n}^\top\beta_{n}\|_{L^2_\Q} =0$, but the infimum is not attained, $\|f - \phi_\theta^\top\beta\|_{L^2_\Q}>0$ for all $(\theta,\beta)\in (\Scal_1)^2\times\R^2$. This proves the non-existence statement.

For the non-uniqueness, assume $f(x)=f(Vx)$ and the pushforward $V_\ast\Q=\Q$ for some $dT\times dT$-matrix $V$. Then, for any $(A,b)\in (\Scal_{dT})^m$ and $\beta\in\R^m$, we have $\|f - \phi_{(A,b)}^\top\beta\|_{L^2_\Q} = \| f - \phi_{(V^\top A,b)}^\top\beta\|_{L^2_\Q}$. But $\spn\{\phi_{(A,b),1},\dots,\phi_{(A,b),m}\}\neq \spn\{\phi_{(V^\top A,b),1},\dots,\phi_{(V^\top A,b),m}\}$ in general by Theorem \ref{thmReLUspanNew}. This completes the proof of Theorem \ref{thmNEReLU}.

\section{Comparison of runtimes}\label{app_runtime}

We discuss how long it takes to run the training and prediction phases on each model from Sections~\ref{sec:call_example} and \ref{sec:insurance}. This involves: running the regression on the number of samples indicated on the first column and calculating $\widehat{V}_1(x_i)$ for 1,000,000 validation samples. 

Table~\ref{tab:timing_comparison_call} shows the runtimes for the European call option example from Section~\ref{sec:call_example}. They clearly show the effect of the dimensionality reduction in the computational cost of the replicating martingale method. 

\begin{table}[H]
    \centering
    \small
    \caption{European call, comparison of runtime (in seconds), single core AMD Opteron 6380 }
\begin{tabular}{@{}rrrrrrrr@{}}
\toprule & &
\multicolumn{2}{r}{\textbf{Full Polynomial basis}} &  \textbf{LDR} & \multicolumn{2}{r}{\textbf{Neural Network}} \\
\textbf{\textbf{Samples}} & \textbf{\textbf{Lasso}} &                    Regress-now & Regress-later & Regress-later &             Regress-now & Regress-later \\
\midrule
\multicolumn{1}{l}{\textbf{\textbf{\textbf{T}}=5}}  &   &   &   &   &   &   \\
                     1,000 &                     1.7 &                            0.9 &           1.5 &           1.4 &                     9.5 &          18.4 \\
                     5,000 &                     3.4 &                            0.9 &           2.3 &           1.6 &                    16.1 &          21.1 \\
                    10,000 &                     6.1 &                            0.9 &           3.0 &           2.1 &                    24.0 &          24.4 \\
                    50,000 &                    24.9 &                            1.0 &          13.4 &           7.1 &                    92.6 &          49.4 \\
\multicolumn{1}{l}{\textbf{\textbf{\textbf{T}}=40}}  &   &   &   &   &   &   \\
                     1,000 &                   153.4 &                            0.9 &             &           7.2 &                    10.0 &          18.5 \\
                     5,000 &                   897.2 &                            0.9 &             &          22.3 &                    15.5 &          24.4 \\
                    10,000 &                  2,199.0 &                            1.0 &             &          51.7 &                    24.0 &          30.0 \\
                    50,000 &                       &                            1.4 &             &         271.7 &                    19.1 &          76.1 \\
\bottomrule
\end{tabular}

    \label{tab:timing_comparison_call}
\end{table}

Table~\ref{tab:timing_comparison_var_ann} shows the runtimes for the insurance liability example from Section~\ref{sec:insurance}. As above, unsurprisingly, we find regress-now methods to be faster than regress-later methods. This might partially explain the popularity with practitioners, especially for frequent calculations that do not require high precision. However, for quarterly or annual calculations of regulatory solvency, it seems hard to justify the much higher error rates for the benefit of saving a few minutes of calculations.
 
\begin{table}[ht]
    \centering
    \small
    \caption{Insurance liability, comparison of runtime (in seconds), single core AMD Opteron 6380 }
\begin{tabular}{@{}rrrrrrrr@{}}
\toprule & 
\multicolumn{2}{r}{\textbf{Full Polynomial basis}} &  \textbf{LDR} & \multicolumn{2}{r}{\textbf{Neural Network}} \\
\textbf{\textbf{Samples}} &                    Regress-now & Regress-later & Regress-later &             Regress-now & Regress-later \\
\midrule
\multicolumn{1}{l}{\textbf{\textbf{\textbf{T}}=5}}  &   &   &   &   &   \\
                     1,000 &                            0.8 &           4.3 &          30.4 &                     4.0 &           5.5 \\
                     5,000 &                            0.9 &          54.4 &         113.2 &                    11.1 &           8.2 \\
                    10,000 &                            1.1 &          45.3 &          49.6 &                    19.8 &          11.9 \\
                    50,000 &                            3.0 &         108.1 &         238.6 &                    95.1 &          42.7 \\
\multicolumn{1}{l}{\textbf{\textbf{\textbf{T}}=40}}  &   &   &   &   &   \\
                     1,000 &                            1.3 &             &         279.1 &                     4.6 &           6.3 \\
                     5,000 &                            2.6 &             &         698.9 &                    12.9 &          15.7 \\
                    10,000 &                            4.3 &             &        1,815.1 &                    22.9 &          26.5 \\
                    50,000 &                           17.9 &             &        1,472.8 &                    31.7 &         115.5 \\
\bottomrule
\end{tabular}

    \label{tab:timing_comparison_var_ann}
\end{table}

The slowest method is the polynomial LDR, which for the high dimensional problem can take take up to 25 minutes to find the solution and make the estimation of the out-of-sample distribution. This time is entirely dominated by the optimization---training---step, not the estimation---prediction---step. The polynomial LDR method runtime is extremely sensitive to the $p$ parameter. For example, for $T=40$ and sample size 1,000, it takes 33 seconds to solve with $p=5$ and 279.1 seconds to solve with $p=10$---the latter is the example shown in Table \ref{tab:timing_comparison_var_ann}.

The neural network model can be solved relatively fast, taking 2 minutes in the largest problem.

\section{Sensitivity analysis to hyper-parameters}\label{sec:sensitivity}

We discuss the sensitivity of our results with respect to the choice of hyper-parameters for the polynomial LDR and the neural network. 

\subsection{Sensitivity of polynomial LDR} \label{subsec:sensitivity_ldr}
The polynomial LDR method in Section~\ref{sec:ldr_part1} has two hyper-parameters, the target dimensionality $p$ and the polynomial degree $\delta$. Additionally, the Riemannian BFGS algorithm used to solve the optimization problem adds several other parameters, the main one being the starting point for the parameter $A$, called here $A_0$.

The polynomial degree parameter is common to all polynomial approximations, and has the expected impact on the results. In this section, we focus on the parameter $p$ which is unique to the linear dimensionality reduction and the parameter $A_0$ which in our empirical examples proved to have a large impact on results.

We show that the choice of $p$ is purely a trade-off between approximation error and number of samples required, and that the choice of starting point makes a very large difference in the final results. A random starting point performs relatively badly, compared to a starting point that takes into account the fact that in financial models, cash flows closer in time are usually more important than those farther in time.

\subsubsection{Starting point $A_0$}
The Riemannian BFGS algorithm used to solve the polynomial LDR optimization problem requires a starting point for $A$.

A first, simple way of generating a starting point---similar to what is done for the L-BFGS algorithm used to solve the neural network optimization problem---is to generate it randomly. To do this we draw $dTp$ random samples from $N(0,1)$ and arrange them into an $dT \times p$ matrix $B$. Then $A_0=B(B^\top B)^{-1/2}$ is a random matrix that follows the uniform distribution on the Stiefel manifold $V_p(\R^{dT})$.

A second way is to use a rectangular diagonal matrix and fill the last column to ensure that every one of the $dT$ input dimensions has a weight in at least one of the $p$ output dimension, that is $A_0 = B \mid B_{ij} = 1 \text{ if } (i=j \land i\neq p) \land B_{ip}=\frac{1}{\sqrt{dT-p+1}} \text{ if } i \geqslant p$. Conceptually, this starting point can be thought as a point where those input dimensions farthest in the future have been grouped into one output dimension. The following is an example for $dT=4$ and $p=3$:

\[  \begin{pmatrix}1&0&0\\0&1&0\\0&0&\sqrt{0.5}\\0&0&\sqrt{0.5}\end{pmatrix} . \]

A third way uses the same rationale of grouping input dimensions that are far in the future into one output dimension, but does so respecting the fact that $X \in \R^{T\times d}$ and only groups variables across time ($T$) but not across dimensions ($d$). The following is an example for $T=5$, $d=3$ and $p=3$ which corresponds to what was used in the European call example in Section~\ref{sec:call_example}:

\[ \begin{pmatrix}\sqrt{T}&0&0\\0&\sqrt{T}&0\\0&0&\sqrt{T}\\
\vdots & \vdots & \vdots \\
\sqrt{T}&0&0\\0&\sqrt{T}&0\\0&0&\sqrt{T}\end{pmatrix}  .\]

Since this method, which we call \say{folding}, provides the best results we also use it in Sections~\ref{sec:call_example} and \ref{sec:insurance}. In the latter we work with $T=5$, $d=5$ and $p=10$, and leads to a starting point (in block notation):

\begin{center}
\begin{tikzpicture}[decoration=brace]
    \matrix (m) [matrix of math nodes,left delimiter=(,right delimiter={)}] {
    \mathbf{1}_d & \mathbf{0}_{d\times(p-d)} \\
    \mathbf{0}_{d\times d} & \sqrt{T-1}\mathbf{1}_{p-d}\\\
    \vdots & \vdots \\
    };
    \draw[decorate, transform canvas={xshift=1.3em},thick] (m-2-2.east|-m-2-2.north east) -- node[right=2pt] {$T-1$ times} (m-2-2.east|-m-2-2.south east);
\end{tikzpicture}.
\end{center}

In Tables \ref{tab:varann_es_init_methods} and \ref{tab:varann_l1_init_methods}, which correspond to Tables~\ref{tab:varann_mae_es_all_methods} and \ref{tab:varann_l1_error_all_methods}, we can see the comparison across different starting points. The folding starting point performs best of all starting points. This is not surprising since it is expected that for this type of models---European call option and insurance liability---the combination of path dependency and discounting makes variables closer in time relatively more important than those farther in time. A disappointing characteristic revealed in the data is that when increasing the number of training samples, we do not always get a strictly decreasing error. In fact, the error seems to stabilize relatively early---around 5000 samples---and then only be subject to small fluctuations. The comparison across starting points confirms that this lack of improvement is not due to a lack of a better solution, but rather most likely to the presence of local minima.
 
\begin{table}[H] 
    \centering
    \small
    \caption{Insurance liability, comparison of expected shortfall MApE (in percentage points) for different starting points in polynomial LDR basis}
\begin{tabular}{@{}rrrrrrrr@{}}
\toprule
\textbf{Samples} & \textbf{Folding} & \textbf{Diagonal} & \textbf{Random} \\
\midrule
\multicolumn{1}{l}{\textbf{\textbf{T}=5}}  &   &   &   \\
           1,000 &              7.9 &              14.9 &            16.6 \\
           5,000 &              5.3 &              17.5 &            31.5 \\
          10,000 &              5.1 &              23.6 &            31.8 \\
          50,000 &              3.7 &              31.0 &            32.4 \\
\multicolumn{1}{l}{\textbf{\textbf{T}=40}}  &   &   &   \\
           1,000 &             22.9 &              38.0 &            63.9 \\
           5,000 &              5.7 &              54.7 &            69.4 \\
          10,000 &              5.9 &              53.0 &            68.2 \\
          50,000 &              7.2 &              49.2 &            67.0 \\
\bottomrule
\end{tabular}

    \label{tab:varann_es_init_methods}
\end{table}

\begin{table}[H] 
    \centering
    \small
    \caption{Insurance liability, comparison of relative mean $L_1$ error (in percentage points) for different starting points in polynomial LDR basis}
\begin{tabular}{@{}rrrrrrrr@{}}
\toprule
\textbf{Samples} & \textbf{Folding} & \textbf{Diagonal} & \textbf{Random} \\
\midrule
\multicolumn{1}{l}{\textbf{\textbf{T}=5}}  &   &   &   \\
           1,000 &              0.2 &               1.0 &             1.2 \\
           5,000 &              0.2 &               0.8 &             1.0 \\
          10,000 &              0.2 &               0.9 &             1.0 \\
          50,000 &              0.2 &               0.9 &             1.0 \\
\multicolumn{1}{l}{\textbf{\textbf{T}=40}}  &   &   &   \\
           1,000 &              0.8 &               1.7 &             1.7 \\
           5,000 &              0.3 &               1.3 &             1.5 \\
          10,000 &              0.4 &               1.3 &             1.5 \\
          50,000 &              0.4 &               1.2 &             1.5 \\
\bottomrule
\end{tabular}

    \label{tab:varann_l1_init_methods}
\end{table}

\subsubsection{Target dimensionality parameter $p$}

To show the effects of parameter $p$ on the insurance example, we choose one of the starting point methods (diagonal) and one maturity ($T=5$). The results in Tables \ref{tab:varann_es_ldr_param} and \ref{tab:varann_l1_ldr_param}, which correspond to Tables~\ref{tab:varann_mae_es_all_methods} and \ref{tab:varann_l1_error_all_methods}, confirm the expected effect of changing this parameter: larger values of $p$ produce better results (since the feature map $\phi_\theta$ is a richer function) but also require more training samples to do so. We can see that when moving from $p=10$---used in the main results for the insurance example---to $p=15$ and therefore from $m=286$ to $m=816$ the error for 1,000 training samples increases by a factor of 10 in the expected shortfall and by a factor of 3 in the $L_1$ metric. In those cases with more training samples---5,000 and above---the error goes down as expected.

\begin{table}[H]
    \centering
    \small
    \caption{Insurance liability, comparison of expected shortfall MApE (in percentage points) for different values of the target dimensionality parameter in polynomial LDR basis}
\begin{tabular}{@{}rrrrrrrr@{}}
\toprule
\textbf{Samples} & \textbf{Diagonal p=5} & \textbf{Diagonal p=10} & \textbf{Diagonal p=15} & \textbf{Diagonal p=20} \\
\midrule
\multicolumn{1}{l}{\textbf{\textbf{T}=5}}  &   &   &   &   \\
           1,000 &                  22.3 &                   14.9 &                  228.0 &                  767.9 \\
           5,000 &                  29.3 &                   17.5 &                    8.2 &                    9.7 \\
          10,000 &                  33.3 &                   23.6 &                   13.2 &                    9.0 \\
          50,000 &                  37.5 &                   31.0 &                   21.0 &                   10.7 \\
\bottomrule
\end{tabular}

    \label{tab:varann_es_ldr_param}
\end{table}

\begin{table}[H]
    \centering
    \small
    \caption{Insurance liability, comparison of relative mean $L_1$ error (in percentage points) for different values of the target dimensionality parameter in polynomial LDR basis}
\begin{tabular}{@{}rrrrrrrr@{}}
\toprule
\textbf{Samples} & \textbf{Diagonal p=5} & \textbf{Diagonal p=10} & \textbf{Diagonal p=15} & \textbf{Diagonal p=20} \\
\midrule
\multicolumn{1}{l}{\textbf{\textbf{T}=5}}  &   &   &   &   \\
           1,000 &                   1.0 &                    1.0 &                    3.1 &                   37.3 \\
           5,000 &                   1.0 &                    0.8 &                    0.6 &                    0.5 \\
          10,000 &                   1.0 &                    0.9 &                    0.6 &                    0.6 \\
          50,000 &                   1.0 &                    0.9 &                    0.7 &                    0.5 \\
\bottomrule
\end{tabular}

    \label{tab:varann_l1_ldr_param}
\end{table}

\subsection{Sensitivity of neural network}  \label{subsec:sensitivity_nn}

The neural network method has one main hyper-parameter, the width of the hidden layer. Other typical neural network hyper-parameters as number of layers or activation function do not apply in this case, since the closed-form of the time-t expectation has been defined only for single-layer, ReLu networks. Unlike the polynomial LDR basis, we do not explore the impact of the starting point, since a random starting point already performs very well.

In the main results in Section~\ref{sec:insurance}, we use the same layer width in all cases, $p=100$. This value is the results of a sensitivity test done for different values (10, 50, 100, 200) after which we chose the best results overall cases. This sensitivity test is similar to cross validation but is performed on entirely out-of-sample data, rather than partitioning the existing training data. This has the advantage of keeping the full size of the sample for each regression instead of having to reduce it to allow a percentage to be used as validation set. While cross validation is more frequently used when the total sample budget is fixed, sensitivity analysis is more adequate when one has the ability to generate as many out-of-sample sets as needed. In a different practical setting as the one in this paper, it might be more appropriate to use cross validation for the selection of hyper-parameters.

Using a single choice of layer width in all cases has the advantage of showing good overall results (for different maturities and training sample size) but the disadvantage of being neither optimized for each single case (meaning that the results could have been better when looking at each cell of the table) nor comparable to the polynomial method in terms of functional complexity, that is, the number of parameters that describe the feature map $\phi_\theta$.

The selection could have been done in different ways, and in this section we show some alternatives and their effects on the results shown in Section~\ref{sec:insurance}. The results are summarized in Tables~\ref{tab:varann_es_layer_width} and \ref{tab:varann_l1_layer_width}, which correspond to Tables~\ref{tab:varann_mae_es_all_methods} and \ref{tab:varann_l1_error_all_methods}. We show that some of the alternatives perform even better than our choice for the main results, implying the potential for improvement in the neural network basis, which is already the best performing basis in our comparisons.

\begin{table}[H]
    \centering
    \small
    \caption{Insurance liability, comparison of expected shortfall MApE (in percentage points) for different layer widths in neural network basis}
\begin{tabular}{@{}rrrrrrrr@{}}
\toprule
\textbf{Samples} & Fixed $p=100$ & Minimum width & Equal param dims & Equal $m$ \\
\midrule
\multicolumn{1}{l}{\textbf{\textbf{T}=5}}  &   &   &   &   \\
           1,000 &           2.9 &           2.2 &              3.1 &       3.9 \\
           5,000 &           0.5 &           3.2 &              4.1 &       0.4 \\
          10,000 &           0.6 &           3.3 &              4.2 &       0.2 \\
          50,000 &           0.8 &           3.3 &              4.2 &       0.3 \\
\multicolumn{1}{l}{\textbf{\textbf{T}=40}}  &   &   &   &   \\
           1,000 &          14.4 &          14.5 &             17.5 &      14.1 \\
           5,000 &          11.1 &          12.5 &              2.7 &      12.8 \\
          10,000 &          10.5 &          12.6 &              2.0 &      13.2 \\
          50,000 &           0.5 &           1.4 &              2.2 &       2.4 \\
\bottomrule
\end{tabular}

    \label{tab:varann_es_layer_width}
\end{table}

\begin{table}[H]
    \centering
    \small
    \caption{Insurance liability, comparison of relative mean $L_1$ error (in percentage points) for different layer widths in neural network basis}
\begin{tabular}{@{}rrrrrrrr@{}}
\toprule
\textbf{Samples} & Fixed $p=100$ & Minimum width & Equal param dims & Equal $m$ \\
\midrule
\multicolumn{1}{l}{\textbf{\textbf{T}=5}}  &   &   &   &   \\
           1,000 &           0.1 &           0.1 &              0.1 &       0.1 \\
           5,000 &           0.1 &           0.1 &              0.1 &      \textless{}0.1 \\
          10,000 &          \textless{}0.1 &           0.1 &              0.1 &      \textless{}0.1 \\
          50,000 &          \textless{}0.1 &           0.1 &              0.1 &      \textless{}0.1 \\
\multicolumn{1}{l}{\textbf{\textbf{T}=40}}  &   &   &   &   \\
           1,000 &           6.1 &           6.4 &              2.8 &       6.5 \\
           5,000 &           0.6 &           0.8 &              0.5 &       0.8 \\
          10,000 &           0.5 &           0.5 &              0.2 &       0.5 \\
          50,000 &           0.1 &           0.1 &              0.1 &       0.1 \\
\bottomrule
\end{tabular}

    \label{tab:varann_l1_layer_width}
\end{table}

The first alternative is to use the theoretical minimum width for the network, as described in \cite{hanin2017approximating}. In our case, it means using $p=25$ for $T=5$ and $p=200$ for $T=40$. This method does not show a good performance. Interestingly, it performs worse even for $T=40$ where $p=100$ is below the minimum. Still, this is not a violation of the theoretical minimum since it assumes a neural network of arbitrary depth, so it is always possible that using more hidden layers would result in smaller errors than the fixed $p$ method.

The second alternative is to backsolve the width of the network that creates a parameter space of similar dimensionality as that of the polynomial LDR method. For $d=5$, $T=5$ and $p=10$ the polynomial LDR has 481 parameters. For $T=40$ it has 2,231 parameters. This can be matched by using a neural network with $p=18$ and $p=11$ nodes respectively. This alternative provide very good results for the high dimensionality case ($T=40$) but not as good for the low dimensionality case ($T=5$).

The third and final alternative is to use a neural network that matches the number of basis functions $m$. This means, for both $T=5$ and $T=40$, that $p=286$. The results for this alternative are similar to the other alternatives.

\bibliographystyle{alpha}
\bibliography{drpbib}

\end{document}